\documentclass[a4paper,USenglish,cleveref,autoref,thm-restate]{lipics-v2021}

 \newif\iflong
\newif\ifshort

\longtrue

\iflong
\else
 \shorttrue
\fi

\usepackage[noadjust]{cite}
 
 \usepackage{amsmath,stmaryrd,graphicx}

\makeatletter
\newcommand{\fixed@sra}{$\vrule height 2\fontdimen22\textfont2 width 0pt\shortrightarrow$}
\newcommand{\shortarrow}[1]{   \mathrel{\text{\rotatebox[origin=c]{\numexpr#1*45}{\fixed@sra}}}
}
\newcommand{\downleftarrow}[0]{\shortarrow{5}}
\newcommand{\downrighttarrow}[0]{\shortarrow{7}}

\usepackage{verbatim}
\usepackage[textsize=footnotesize]{todonotes}
\usepackage{xcolor}
\usepackage[textsize=footnotesize]{todonotes}

\usepackage{tcolorbox}
\usepackage{cite}

\usepackage{etex}
\usepackage{enumerate}
\usepackage{graphics}
\usepackage{xspace}
 
\usepackage{caption}
\usepackage{subcaption}
\usepackage{epsfig}
\usepackage{arcs}
\usepackage[T1]{fontenc}
\usepackage{alltt}
\usepackage{amssymb,amsfonts,epsf,url}
\usepackage{graphicx}
\usepackage{pstricks,pstricks-add}
\usepackage{pst-node}
\usepackage{pst-coil}
\usepackage{amsmath}
\usepackage{amsthm}
\usepackage{verbatim}
\usepackage[textsize=footnotesize]{todonotes}
\usepackage{xcolor}
\usepackage[textsize=footnotesize]{todonotes}
\usepackage{fancyhdr}
\usetikzlibrary{arrows}

\newcommand{\cmpl}{\textsc{\textup{GCMP}}\xspace} 
\newcommand{\cmplone}{\textsc{\textup{GCMP1}}\xspace}

\theoremstyle{plain}
 
\RequirePackage{etex}

\newcommand{\bigoh}{\mathcal{O}}

\newtheorem{longtheorem}{Theorem}
\newtheorem{longlemma}[longtheorem]{Lemma}
\newtheorem{longproposition}[longtheorem]{Proposition}
\newtheorem{longdefinition}[longtheorem]{Definition}

\theoremstyle{definition}
\newtheorem{our_definition}[theorem]{Definition}

\newlength{\alginputwidth}

\newcommand{\Oh}{{\mathcal O}}
\newcommand{\nat}{\mathbb{N}}
\newcommand{\FFF}{\mathcal{F}}
\newcommand{\F}{\mathcal{F}}
\newcommand{\M}{\mathcal{M}}

\newcommand{\Pol}{\mbox{\sf P}}

\newcommand{\NP}{\mbox{{\sf NP}}}
\newcommand{\tw}{\mbox{{\sf tw}}}

\newcommand{\FPT}{\mbox{{\sf FPT}}}

\newcommand{\W}{\mbox{{\sf W}}}

\newcommand{\R}{\mathcal{R}}
\newcommand{\SSS}{\mathbf{\emph{S}}}

\newcommand{\dist}{\emph{\textnormal{dist}}}

\newcommand{\I}{\mathcal{I}}

\newenvironment{tightcenter}
 {\parskip=0pt\par\nopagebreak\centering}
 {\par\noindent\ignorespacesafterend}

\usepackage{tikz}
\usetikzlibrary{calc}
\usepackage{xargs}
\usepackage{xifthen}
\usepackage{framed}

\usepackage{ctable}

\newlength{\RoundedBoxWidth}
\newsavebox{\GrayRoundedBox}
\newenvironment{GrayBox}[1]%
   {\setlength{\RoundedBoxWidth}{\textwidth-10.5ex}
    \def\boxheading{#1}
    \begin{lrbox}{\GrayRoundedBox}
       \begin{minipage}{\RoundedBoxWidth}%
   }{%
       \end{minipage}
    \end{lrbox}%
    \begin{tightcenter}%
    \begin{tikzpicture}%
       \node(Text)[draw=black!90,fill=white,rounded corners,%
             inner sep=2ex,text width=\RoundedBoxWidth]%
             {\usebox{\GrayRoundedBox}};
        \coordinate(x) at (current bounding box.north west);
        \node [draw=white,rectangle,inner sep=3pt,anchor=north west,fill=white] 
        at ($(x)+(6pt,.75em)$) {\boxheading};
    \end{tikzpicture}
    \end{tightcenter}\vspace{0pt}%
    \ignorespacesafterend
}    

\newenvironment{problem}[2][]{\noindent\ignorespaces%
                                \FrameSep=8pt%
                                \parindent=0pt%
                \vspace*{-.5em}
                \ifthenelse{\isempty{#1}}{%
                  \begin{GrayBox}{\textsc{#2}}%
                }{%
                }
                \newcommand\Prob{{Problem:}}%
                \newcommand\Input{{Input:}}%
                          
                \begin{tabular*}{\textwidth}{@{\hspace{.1em}} >{\itshape} p{1.2cm} p{0.85\textwidth} @{}}%
            }{
                \end{tabular*}%
                \end{GrayBox}%
                \vspace*{-.5em}
                \ignorespacesafterend
            } 
 
       \title{Parameterized Algorithms for Coordinated Motion Planning: Minimizing Energy}
  \titlerunning{Parameterized Algorithms for Coordinated Motion Planning: Minimizing Energy}

\author{Argyrios Deligkas}{Department of Computer Science, Royal Holloway, University of London, Egham, UK}{Argyrios.Deligkas@rhul.ac.uk}{https://orcid.org/0000-0002-6513-6748}{Supported by Engineering and Physical Sciences Research Council (EPSRC) grant EP/X039862/1.}

\author{Eduard Eiben}{Department of Computer Science, Royal Holloway, University of London, Egham, UK}{eduard.eiben@rhul.ac.uk}{https://orcid.org/0000-0003-2628-3435}{}

\author{Robert Ganian}{Algorithms and Complexity Group, TU Wien, Vienna, Austria}{rganian@gmail.com}{https://orcid.org/0000-0002-7762-8045}{Project No. Y1329 of the Austrian Science Fund (FWF), Project No. ICT22-029 of the Vienna Science Foundation (WWTF).}

\author{Iyad Kanj}{School of Computing, DePaul University, Chicago, USA}{ikanj@cdm.depaul.edu}{0000-0003-1698-8829}{DePaul URC Grants 606601 and 350130}
\author{M. S. Ramanujan}{Department of Computer Science, University of Warwick, UK}{R.Maadapuzhi-Sridharan@warwick.ac.uk}{https://orcid.org/0000-0002-2116-6048}{ Supported by Engineering and Physical Sciences Research Council (EPSRC) grants EP/V007793/1 and EP/V044621/1.}

\authorrunning{A.\ Deligkas, E.\ Eiben, R.\ Ganian, I.\ Kanj, R.\ M.\ Sridharan}  
\Copyright{A.\ Deligkas, E.\ Eiben, R.\ Ganian, I.\ Kanj, R.\ M.\ Sridharan} 
 

\relatedversion{}

\begin{document}
\nolinenumbers

\maketitle
 
\begin{abstract}
We study the parameterized complexity of a generalization of the coordinated motion planning problem on graphs, where the goal is to route a specified subset of a given set of $k$ robots to their destinations with the aim of minimizing the total energy (i.e., the total length traveled). We develop novel techniques to push beyond previously-established results that were restricted to solid grids.

 We design a fixed-parameter additive approximation algorithm for this problem parameterized by $k$ alone. This result, which is of independent interest, allows us to prove the following two results pertaining to well-studied coordinated motion planning problems: (1) A fixed-parameter algorithm, parameterized by $k$, for routing a single robot to its destination while avoiding the other robots, which is related to the famous Rush-Hour Puzzle; and (2) a fixed-parameter algorithm, parameterized by $k$ plus the treewidth of the input graph, for the standard \textsc{Coordinated Motion Planning} (CMP) problem in which we need to route all the $k$ robots to their destinations. 
The latter of these results implies, among others, the fixed-parameter tractability of CMP parameterized by $k$ on graphs of bounded outerplanarity, which include bounded-height subgrids.

We complement the above results with a lower bound which rules out the fixed-parameter tractability for CMP when parameterized by the total energy. This contrasts the recently-obtained tractability of the problem on solid grids  under the same parameterization. As our final result, we strengthen the aforementioned fixed-parameter tractability to hold not only on solid grids but all graphs of bounded local treewidth -- a class including, among others, all graphs of bounded genus.

     \end{abstract}
  
  \section{Introduction}
The task of routing a set of robots (or agents) from their initial positions to specified destinations while avoiding collisions is of great importance in a multitude of different settings. Indeed, there is an extensive body of works dedicated to algorithms solving this task, especially in the computational geometry~\cite{AdlerBHS15,hopcroft,hopcroft1984,LiuSZZ21,reif,solovey,sharir,sharir1,sharir2}, artificial intelligence~\cite{heuristic3,heuristic1,heuristic4} and robotics~\cite{banfi,heuristic2,heuristic5} research communities. Such algorithms typically aim at providing a schedule for the robots which is not only safe (in the sense of avoiding collisions), but which also optimizes a certain measure of the schedule -- typically its makespan or energy (i.e., the total distance traveled by all robots). In this article, we focus solely on the task of optimizing the latter of these two measures.

A common formalization of our task of interest is given by the \textsc{Coordinated Motion Planning} problem (also known as \textsc{Multiagent Pathfinding}). In the context of energy minimization, the problem can be stated as follows: Given a graph $G$, a budget $\ell$ and a set $\M$ of robots, each equipped with an initial vertex and destination vertex in $G$, compute a schedule which delivers all the robots to their destinations while ensuring that the combined length traveled of all the robots is at most $\ell$. Here, a schedule can be seen as a sequence of commands to the robots, where at every time step each of the robots can move to an adjacent vertex as long as no vertex or edge is used by more than a single robot at that time step. \textsc{Coordinated Motion Planning}, which generalizes the famous \NP-hard $(n^2-1)$-puzzle~\cite{nphard2,nphard1}, has been extensively studied---both in the discrete setting considered here as well as in various continuous settings~\cite{heuristic3,heuristic2,heuristic1,heuristic5,heuristic4,lavalle,lavalle1,halprinunlabeled,demaine,survey,alagar,sharir, sharir1, sharir2}---
and has also been the target of specific computing challenges~\cite{socg2021}.
  In other variants of the problem, there is no requirement to route \emph{all} the robots to their destinations -- sometimes the majority (or all but $1$) of the robots are simply movable obstacles that do not have destinations of their own. This is captured, e.g., by the closely-related and well-studied \textsc{Rush Hour} puzzle/problem~\cite{FlakeB02,BrunnerCDHHSZ21} and by \textsc{Graph Motion Planning with 1 Robot} (\textsc{GCMP1})~\cite{PapadimitriouRST94}, which both feature a single robot with a designated destination.
  
Both \textsc{Coordinated Motion Planning} and \textsc{GCMP1} are known to be \NP-hard~\cite{nphard1,nphard2,lavalle}. While several works have already studied \textsc{Coordinated Motion Planning} in the context of approximation and classical complexity theory, a more fine-grained investigation of the difficulty of finding optimal schedules through the lens of \emph{parameterized complexity}~\cite{DowneyFellows13,CyganFKLMPPS15} was only carried out recently~\cite{EibenGK23,KnopEtalMAPF24}. 
 The work in~\cite{EibenGK23} established the fixed-parameter tractability\footnote{A problem is \emph{fixed-parameter tractable} w.r.t.\ a parameter $k$ if it can be solved in time $f(k)\cdot n^{\bigoh(1)}$, where $n$ is the input size and $f$ is a computable function.} of \textsc{Coordinated Motion Planning} parameterized by either the number $k$ of robots or the budget $\ell$ (as well as the makespan variant when parameterized by $k$), but only on solid grids; a \emph{solid} grid is a standard rectangular $p \times q$ grid (i.e., with no holes), for some $p, q \in \nat$. The more recent work in~\cite{KnopEtalMAPF24} showed the \W[1]-hardness of the makespan variant of \textsc{Coordinated Motion Planning} w.r.t.\ the number of robots. The paper~\cite{KnopEtalMAPF24} also showed the \NP-hardness of the makespan problem-variant on trees, and presented parameterized complexity results with respect to several combinations of parameters.  
In this article, we focus our attention on \textsc{GCMP}~\cite{PapadimitriouRST94}, which generalizes both \textsc{Coordinated Motion Planning} and \textsc{GCMP1} by allowing an arbitrary partitioning of the robots into those with destinations and those which act merely as movable obstacles\footnote{While previous hardness results for \textsc{GCMP} considers serial motion of robots (e.g., ~\cite{PapadimitriouRST94}), the hardness applies to parallel motion as well. Here we obtain algorithms for the coordinated motion variant with parallel movement, but note that the results can be directly translated to the serial version.}.
 
\smallskip
\noindent
\textbf{Contribution.}\quad
The aim of this article is to push our understanding of the parameterized complexity of finding energy-optimal schedules beyond the class of solid grids. While this aim was already highlighted in the aforementioned paper on solid grids~\cite[Section 5]{EibenGK23}, the techniques used there are highly specific to that setting and it is entirely unclear how one could generalize them even to the setting of subgrids; a \emph{subgrid} is a subgraph of a solid grid and we define its \emph{height} to be the minimum of its two dimensions.

 As our two main contributions, we provide novel fixed-parameter algorithms (1) for \textsc{GCMP1} parameterized by the number of robots alone (Theorem~\ref{thm:fptrushhour}), and (2) for \textsc{GCMP} parameterized by the number of robots plus the treewidth of the graph (Theorem~\ref{thm:boundedTreewidthAlgorithm}). Theorem~\ref{thm:boundedTreewidthAlgorithm} implies, as a corollary, the fixed-parameter tractability of \cmpl\ parameterized by $k$ plus the minimum dimension of the subgrid.
  
\iflong
\begin{restatable}{theorem}{fptrushhour}
\label{thm:fptrushhour}
  \cmplone{} is \FPT{} parameterized by the number $k$ of robots.
\end{restatable}
\fi
\ifshort
\begin{restatable}{theorem}{fptrushhour}
\label{thm:fptrushhour}
  \cmplone{} is \FPT{} parameterized by the number $k$ of robots.
\end{restatable}
\fi

\begin{restatable}{theorem}{fpttw}
\label{thm:boundedTreewidthAlgorithm}
	{\cmpl} is \FPT\ parameterized by the number of robots and the treewidth of the input graph. 
\end{restatable}

The main technical tool we use to obtain both of these results is a novel fixed-parameter approximation algorithm for \textsc{Generalized Coordinated Motion Planning} parameterized by $k$ alone, where the approximation error is only \emph{additive} in $k$. We believe this result -- summarized in Theorem~\ref{thm:fptapxcorollary} below -- to be of independent interest.

\iflong
\begin{restatable}{theorem}{fptapxcorollary}
\label{thm:fptapxcorollary}
 There is an \FPT{} approximation algorithm for \cmpl{} parameterized by the number $k$ of robots which guarantees an additive error of $\Oh(k^5)$. 
\end{restatable}
\fi
\ifshort
\begin{restatable}{theorem}{fptapxcorollary}
\label{thm:fptapxcorollary}
 There is an \FPT{} approximation algorithm for \cmpl{} parameterized by the number $k$ of robots which guarantees an additive error of $\Oh(k^5)$. 
\end{restatable}
\fi

 The proof of Theorem~\ref{thm:fptrushhour} builds upon Theorem~\ref{thm:fptapxcorollary}. 
For the proof of Theorem~\ref{thm:boundedTreewidthAlgorithm}, we need to combine the approximation algorithm with 
novel insights concerning the ``decomposability'' of schedules along small separators in order to design a treewidth-based dynamic programming algorithm for the problem. A brief summary of the ideas used in this proof is provided at the beginning of Section~\ref{sec:fptapx}.

We complement our positive results which use $k$ as a parameter with an algorithmic lower bound showing that \textsc{Coordinated Motion Planning} is \W[1]-hard (and hence not fixed-parameter tractable under well-established complexity-theoretic assumptions) when parameterized by the energy budget $\ell$. This result (Theorem~\ref{thm:wone} below) contrasts the fixed-parameter tractability of the problem under the same parameterization when restricted to solid grids~\cite[Theorem 19]{EibenGK23}. 

\begin{restatable}
{theorem}{wone}
\label{thm:wone}
 \cmpl\ is \W\textnormal{[1]}-hard when parameterized by $\ell$.
\end{restatable}

While Theorem~\ref{thm:wone} establishes the intractability of the problem when parameterized by the energy on general graphs, one can in fact show that {\cmpl} is fixed-parameter tractable under the same parameterization when the graphs are ``well-structured'' in the sense of having bounded \emph{local treewidth}~\cite{Eppstein00,Grohe03_localtw}. This implies fixed-parameter tractability, e.g., on graphs of bounded genus and generalizes the aforementioned result on grids~\cite[Theorem 19]{EibenGK23}.

\begin{restatable}{theorem}{fptenergy}
\label{thm:FPTbyEnergyOnBoundedLocalTreewidth}
	{\cmpl} is \FPT\ parameterized by $\ell$ on graph classes of bounded local treewidth.
\end{restatable}

\subparagraph{Further Related Work.}
\label{subsec:relatedwork}
        As surveyed above, the computational complexity of \cmpl has received significant attention, particularly by researchers in the fields of computational geometry, AI/Robotics, and theoretical computer science. The problem has been shown to remain \NP-hard under a broad set of restrictions, including on graphs where only a single vertex is not occupied~\cite{goldreich2011finding}, on grids~\cite{banfi}, and bounded-height subgrids~\cite{GeftHalperin}. On the other hand, a feasibility check for the existence of a schedule can be carried out in polynomial time~\cite{YuR14}. The recent AAMAS blue sky paper~\cite{salzman2020research} also highlighted the need of understanding the hardness of the problem and asked for a deeper investigation of its parameterized complexity.

\section{Terminology and Problem Definition}
The graphs considered in this paper are undirected simple graphs. We assume familiarity with the standard graph-theoretic concepts and terminology~\cite{Diestel12}. For a subgraph $H$ of a graph $G$ and two vertices $u, v \in V(H)$, we denote by $\dist_H(u, v)$ the length of a shortest path in $H$ between $u$ and $v$. We write $[n]$, where $n \in \nat$, for $\{1, \ldots, n\}$.

\label{sec:prelim}
\ifshort
We also assume basic familiarity with parameterized complexity theory, including \emph{fixed-parameter tractability}, \emph{parameterized reductions}, and the class \W[1]
 \cite{DowneyFellows13,CyganFKLMPPS15}. 
\fi
\iflong
\subparagraph*{Parameterized Complexity.}
 A {\it parameterized problem} $Q$ is a subset of $\Omega^* \times
\nat$, where $\Omega$ is a fixed alphabet. Each instance of $Q$ is a pair $(I, \kappa)$, where $\kappa \in \nat$ is called the {\it
parameter}. A parameterized problem $Q$ is
{\it fixed-parameter tractable} (\FPT)~\cite{CyganFKLMPPS15,DowneyFellows13,FlumGrohe06}, if there is an
algorithm, called an {\em \FPT-algorithm},  that decides whether an input $(I, \kappa)$
is a member of $Q$ in time $f(\kappa) \cdot |I|^{\bigoh(1)}$, where $f$ is a computable function and $|I|$ is the input instance size.  The class \FPT{} denotes the class of all fixed-parameter tractable parameterized problems.

Showing that a parameterized problem is hard for the complexity classes \W[1] or \W[2] rules out the existence of a fixed-parameter algorithm under well-established complexity-theoretic assumptions. Such hardness results are typically established via a \emph{parameterized reduction}, which is an analogue of a classical polynomial-time reduction with two notable distinctions: a parameterized reduction can run in time $f(k)\cdot n^{\bigoh(1)}$, but the parameter of the produced instance must be upper-bounded by a function of the parameter in the original instance. We  refer to~\cite{CyganFKLMPPS15,DowneyFellows13} for more information on parameterized complexity. 
\fi

\subparagraph*{Treewidth.}
Treewidth is a structural parameter that provides a way of expressing the resemblance of a graph to a forest. 
 Formally, the treewidth of a graph is defined via the notion of tree decompositions as follows.

\begin{our_definition}[{\bf Tree decomposition}]\label{def:treewidth}
{\rm A \emph{tree decomposition} of a graph $G$ is a pair $(T,\beta)$ of a tree $T$ and $\beta: V(T) \rightarrow 2^{V(G)}$,
such that: 
\begin{itemize}
\item $\bigcup_{t \in V(T)} \beta(t) = V(G)$, 
\item for any edge $e \in E(G)$, there exists a node $t \in V(T)$ such that both endpoints of $e$ belong to $\beta(t)$, and
\item for any vertex $v \in V(G)$, the subgraph of $T$ induced by the set $T_v = \{t\in V(T): v\in\beta(t)\}$ is a tree.
\end{itemize}
The {\em width} of $(T,\beta)$ is $\max_{v\in V(T)}\{|\beta(v)|\}-1$. The {\em treewidth} of $G$ is the minimum width of a tree decomposition of $G$.}
\end{our_definition}

Let $(T,\beta)$ be a tree decomposition of a graph $G$. We refer to the vertices of the tree $T$ as \emph{nodes}. We always assume that $T$ is a rooted tree and hence, we have a natural parent-child and ancestor-descendant relationship among nodes in $T$. A \emph{leaf} node nor a \emph{leaf} of $T$ is a node  with degree exactly one in $T$ which is different from the root node. All the nodes of $T$ which are neither the root node or a leaf are called \emph{non-leaf} nodes. 
The set $\beta(t)$ is called the \emph{bag} at $t$.
  For two nodes $u,t \in  T$, we say that $u$ is a {\em descendant} of $t$, denoted $u \preceq t$, if $t$ lies on the unique path connecting $u$ to the root. Note that every node is its own descendant. If $u\preceq t$ and $u\neq t$, then we write $u\prec t$.
 For a tree decomposition $(T,\beta)$ we also have a mapping $\gamma:V(T)\to 2^{V(G)}$ defined as $\gamma(t)=\bigcup_{u\preceq t} \beta(u)$.

  We use the following structured tree decomposition in our algorithm. 
  
\begin{our_definition}[{\bf Nice tree decomposition}]\label{def:niceTreeDec}{\rm 
Let $(T,\beta)$ be a tree decomposition of a graph $G$, where $r$ is the root of $T$. The tree decomposition $(T,\beta)$ is called a \emph{nice tree decomposition} if the following conditions are satisfied. 
 \begin{enumerate}
 \item $\beta(r) =\emptyset$ and $\beta(\ell) =\emptyset$ for every leaf node $\ell$ of $T$; and
 \item every non-leaf node (including the root node) $t$ of $T$ is of one of the following types:
 \begin{itemize}
 \item {\bf Introduce node:} The node $t$ has exactly one child $t'$ in $T$ and $\beta(t)=\beta(t') \cup \{v\}$, where $v \notin \beta(t')$.
  \item {\bf Forget node:} The node $t$ has exactly one child $t'$ in $T$ and $\beta(t)=\beta(t') \setminus \{v\}$, where $v \in \beta(t')$.
  \item {\bf Join node:} The node $t$ has exactly two children $t_1$ and $t_2$ in $T$ and $\beta(t)=\beta(t_1) =\beta(t_2)$.
 \end{itemize}
 \end{enumerate}
 }
\end{our_definition}

\ifshort
A graph class closed under vertex and edge deletion is said to have \emph{bounded local treewidth}~\cite{Eppstein00,Grohe03_localtw} if the treewidth of each graph in the class is upper-bounded by a function of its diameter. Examples of classes of bounded local treewidth include, e.g., graphs of bounded genus~\cite{BonnetDGKMST22}.
\fi
 \iflong

Efficient fixed-parameter algorithms are known for computing a nice tree-decomposition of near-optimal width:

\begin{longproposition}[\cite{Korhonen21}]\label{fact:findtw} 
	There exists an algorithm which, given an $n$-vertex graph $G$ and an integer~$k$, in time $2^{\bigoh(k)}\cdot n$ either outputs a nice tree-decomposition of $G$ of width at most $2k+1$ and $\bigoh(n)$ nodes, or determines that $\tw(G)>k$.
\end{longproposition}  

{\noindent\bf Boundaried graphs.} Roughly speaking, a boundaried graph is a graph where some vertices are annotated. A formal definition is as follows. 

\begin{longdefinition}[{Boundaried Graph}]\label{def:boundariedGraph}
{\rm A {\em boundaried graph} is a graph $G$ with a set $X\subseteq V(G)$ of distinguished vertices called {\em boundary vertices}. The set $X$ is the {\em boundary} of $G$. A boundaried graph $(G,X)$  is called  a {\em $p$-boundaried} graph if $|X|\leq p$.}
\end{longdefinition}

\begin{longdefinition}\label{def:boundariedSubgraphs}
{\rm 	A {\em $p$-boundaried subgraph} of  a graph $G$ is a $p$-boundaried graph $(H,Z)$ such that (i) $H$ is a vertex-induced subgraph of $G$ and $Z$ separates $V(H)\setminus Z$ from $V(G)\setminus V(H)$.  }
\end{longdefinition}

 \begin{longdefinition}
 {\rm 
Let $G$ be a graph with a tree decomposition $(T,\beta)$ and let $x\in V(T)$.  
\begin{itemize} \item $G_{x,T}^\uparrow$ denotes the boundaried graph obtained by taking the subgraph of $G$ induced by the set $V(G)\setminus  (\gamma(x)\setminus \beta(x))$, with $\beta(x)$ as the boundary. 
\item Similarly, 
$G_{x,T}^\downarrow$ denotes the boundaried graph obtained by taking the subgraph of $G$ induced by $\gamma(x)$ with $\beta(x)$ as the boundary.
\end{itemize}  
 }
\end{longdefinition}

Notice that if $G$ is a graph with a tree decomposition $(T,\beta)$ of width at most $k$ and $x\in T$, then $G_x^\uparrow$ and $G_x^\downarrow$ are both ($k+1$)-boundaried graphs.

\fi

\subparagraph*{Problem Definition.}
In our problems of interest, we are given an undirected graph $G$ and a set $\R=\{R_1, R_2, \ldots, R_k\}$ of $k$ robots where $\R$ is partitioned into two sets $\M$ and $\F$. Each $R_i \in \M$ has a starting vertex $s_i$ and a destination vertex $t_i$ in $V(G)$ and each $R_i \in \F$ is associated only with a starting vertex $s_i \in V(G)$.
We refer to the elements in the set $\{s_i\mid i\in [k]\}\cup \{t_i\mid R_i\in \M\}$ as {\em terminals}. The set $\M$ contains robots that have specific destinations they must reach, whereas $\F$ is the set of remaining ``free'' robots. We assume that all the $s_i$ are pairwise distinct and that all the $t_i$ are pairwise distinct.  At each time step, a robot may either move to an adjacent vertex, or stay at its current vertex, and robots may move simultaneously.  We use a discrete time frame $[0, t]$, $t \in \nat$, to reference the sequence of moves of the robots and in each time step $x \in [0, t]$ every robot remains stationary or moves.

A \emph{route} for robot $R_i$ is a tuple $W_i=(u_0, \ldots, u_{t})$ of vertices in $G$ such that (i) $u_0=s_i$ and $u_{t}=t_i$ if $R_i \in \M$ and (ii) $\forall j \in [t]$, either $u_{j-1} =u_j$ or $u_{j-1}u_{j} \in E(G)$.
         Put simply, $W_i$ corresponds to a ``walk'' in $G$, with the exception that consecutive vertices in $W_i$ may be identical (representing waiting time steps), in which $R_i$ begins at its starting vertex at time step $0$, and if $R_i \in \M$ then $R_i$ reaches its destination vertex at time step $t$. Two routes $W_i=(u_0, \ldots, u_{t})$ and $W_j=(v_0, \ldots, v_{t})$, where $i \neq j \in [k]$, are \emph{non-conflicting} if (i) $\forall r \in \{0, \ldots, t\}$, $u_r \neq v_r$, and (ii) $\nexists r \in \{0, \ldots, t-1\}$ such that $v_{r+1} =u_r$ and $u_{r+1} =v_r$. Otherwise, we say that $W_i$ and $W_j$ \emph{conflict}. Intuitively, two routes conflict if the corresponding robots are at the same vertex at the end of a time step, or go through the same edge (in opposite directions) during the same time step. 

A \emph{schedule} $\SSS$ for $\R$ is a set of pairwise non-conflicting routes $W_i, i \in [k]$, during a time interval $[0, t]$.  The (\emph{traveled}) \emph{length} of a route (or its associated robot) within $\SSS$ is the number of time steps $j$ such that $u_j\neq u_{j+1}$. The \emph{total traveled length} of a schedule is the sum of the lengths of its routes; this value is often called the \emph{energy} in the literature (e.g., see~\cite{socg2021}).

Using the introduced terminology, we formalize the {\sc Generalized Coordinated Motion Planning with Energy Minimization} (\cmpl) problem below.  

\begin{problem}[]{{ \cmpl}}
\Input & A tuple $(G, \R=(\M,\F), k, \ell)$, where $G$ is a graph, $\R=\{R_i \mid i \in [k]\}$ is a set of robots partitioned into sets $\cal M$ and $\cal F$, where each robot in $\cal M$ is given as a pair of vertices $(s_i, t_i)$  and each robot in $\cal F$ as a single vertex $s_i$, and $k, \ell \in \nat$.\\
 \Prob & Is there a schedule for $\R$ of total traveled length at most $\ell$?
\end{problem}

By observing that the feasibility check of Yu and Rus~\cite{YuR14} transfers seamlessly to the case where some robots do not have destinations, we obtain the following.

\begin{proposition} [\cite{YuR14}]
\label{prop:reconfiguration}
    The existence of a schedule for an instance of \cmpl{} can be decided in linear time. Moreover, if such a schedule exists, then a schedule with total length traveled of $\Oh(|V(G)|^3)$ can be computed in $\Oh(|V(G)|^3)$ time. 
 \end{proposition}

Proposition~\ref{prop:reconfiguration} implies inclusion in \NP, and allows us to assume henceforth that every instance of \cmpl{} is feasible (otherwise, in linear time we can reject the instance). 
     We denote by \cmplone\ the restriction of \cmpl\ to instances where $|\M|=1$.
 We remark that even though \cmpl is stated as as a decision problem, all the algorithms provided in this paper are constructive and can output a corresponding schedule (when it exists) as a witness.

\section{An Additive \FPT{} Approximation for \cmpl}
\label{sec:fptapx}
In this section, we give an \FPT{} approximation algorithm for \cmpl{} with an additive error that is a function of the number $k$ of robots.

We start by providing a high-level, low-rigor intuition for the main result of this section. Let $\I=(G, \R=(\M,\F), k, \ell)$ be an instance of \cmpl{}. Ideally, we would like to route the robots in $\M$ along shortest paths to their destinations, while having the other robots move away only a ``little'' ($\epsilon(k)$-many steps) to unblock the shortest paths for the robots in $\M$. Unfortunately, this might not be possible as it is easy to observe that a free robot might have to travel a long distance in order to unblock the shortest paths of other robots. (For instance, think about the situation where a free robot is positioned in the middle of a long path of degree-2 vertices that the shortest paths traverse.) However, we will show that such a situation could only happen if the blocking robots are positioned in simple graph structures containing a long path of degree-2 vertices. 

We then exploit these simple structures to ``guess'' in \FPT-time, for each robot, a location which it visits in an optimal solution and which is in the vicinity of a safe location, called a ``haven''; this haven is centered around a ``nice'' vertex of degree at least 3, and allows the robot to avoid any passing robot within $\epsilon(k)$ moves. We show how to navigate the robots to these havens optimally in the case of \cmplone{}, and with an $\epsilon(k)$ overhead in the case of \cmpl{}. Moreover, this navigation is well-structured and can be leveraged to show that no robot in an optimal solution will visit the same vertex many times. 
A similar navigation takes place at the end, during the routing of the robots from their havens to their destinations.

Once
we obtain such a reduced instance in which all starting positions and destinations of the
 robots are in havens, we can use our intended strategy to navigate each robot in $\M$ along a shortest path,
 with only an $\epsilon(k)$ overhead, which immediately gives us the approximation result and the property that no robot visits any vertex more than $\epsilon(k)$ times in an optimal solution. 
This latter property about the optimal solution is crucial, as we exploit it later in Section~\ref{sec:treewidth} to design an intricate dynamic programming algorithm over a tree decomposition of the input graph.

In addition, each free robot moves at most $\epsilon(k)$ times in the considered solution for the reduced instance.
This, together with the equivalence between $\I$ and the reduced instance in the case of \cmplone{},  allow us to restrict the movement of the free robots in an optimal solution to only $\epsilon(k)$-many locations, which we use in Section~\ref{sec:rushhour} to obtain the \FPT{} algorithm for \cmplone{}.

 We start by defining the notion of a nice vertex and its haven. 
  
\begin{our_definition}[Nice Vertex]
\label{def:nicevertex}
  A vertex $w\in V(G)$ is \emph{nice} if there exist three connected subgraphs $C_1, C_2, C_3$ of $G$ such that: (i) the pairwise intersection of the vertex sets of any pair of these subgraphs is exactly $w$, and (ii) $|V(C_1)| \geq k+1$, $|V(C_2)| \geq k+1$, and $|V(C_3)| \geq 2$. If $w$ is nice, let $x \in V(C_3)$ be a neighbor of $w$, and define the \emph{haven} $H_w$ of $w$ to be the subgraph  of $G$ induced by the vertices in $\{x\} \cup V(C_1) \cup V(C_2)$ whose distance from $w$ in $H_w$ (i.e., $\dist_{H_v}(w, u)$, for $u \in V(C_1) \cup V(C_2)$) is at most $k$. 
\end{our_definition}

For a set $S \subseteq \R$ of robots and a subgraph $H$, a \emph{configuration} of $S$ w.r.t.~$H$ is an injection $\iota: \  S \longrightarrow V(H)$. The following lemma shows that we can take the robots in a haven from any configuration to any other configuration in the haven, while incurring a total of $\Oh(k^3)$ travel (in the haven) length.

\begin{lemma}
\label{lem:swap}
Let $w$ be a nice vertex, let $C_1, C_2, C_3$ be three subgraphs satisfying conditions (i) and (ii) of Definition~\ref{def:nicevertex}, and let $H_w$ be a haven for $w$. 
Then for any set of robots $S \subseteq \R$ in $H_w$ with current configuration $\iota(S)$, any configuration $\iota'(S)$ with respect to $H_w$ can be obtained from $\iota(S)$ via a sequence of $\Oh(k^3)$ moves that take place in $H_w$.
\end{lemma}

\iflong \begin{proof}
Let $T_1$ (resp.~$T_2$) be a Breadth-First Search (BFS) tree in $H_w$ rooted at $w$ whose vertex set is $V(H_w) \cap V(C_1)$ (resp.~$V(H_w) \cap V(C_2)$). Note that $T_1$ and $T_2$ are well defined since $C_1$ and $C_2$ are connected. Moreover, we have $|V(T_1)|, |V(T_2)| \geq k+1$, and $w$ is the only intersection between the two sets $V(T_1)$ and $V(T_2)$. Let $x \in H_w$ be the neighbor of $w$ in $C_3$. To arrive at configuration $\iota'(S)$ from $\iota(S)$, we perform the following process.

In the first phase of this process, we move all the robots in $S$ from their current vertices to arbitrary distinct vertices in $V(T_1) - \{w\}$. Note that this is possible since $|V(T_1) - \{w\}| \geq k$. Moreover, this is achievable in $\Oh(k^2)$ robot moves since every vertex in $S$ is at distance at most $k$ from $w$. For instance, one way of achieving the above is to move first the robots in $S$ that are in $T_1$, one by one to the deepest unoccupied vertex in $T_1$, then repeat this for the robots in $S$ that are in $T_2$, and finally move the robot on $x$ (if there is such a robot) to $T_1$. 

Let $S_2= \{R \in S \mid \iota'(R) \in V(T_2)\setminus 
\{w\}\}$ (i.e., the set of robots in $S$ whose final destinations w.r.t.~$\iota'(S)$ are in $T_2$ excluding vertex $w$). In the second phase, we route the robots in $S_2$ in reverse order of the depths of their destinations in $T_2$; that is, for two robots $R$ and $R'$ in $S_2$, if $R$'s destination is deeper than that of $R'$, then $R$ is routed first; ties are broken arbitrarily. To route a robot $R \in S_2$, we first move $R$ in $T_1$ towards $w$ as follows. Let $P$ be the unique path in $T_1$ from $w$ to the vertex at which $R$ currently is. (Note that $P$ may contain other robots.) We ``pull'' the robots along $P$ towards $w$, and into $T_2$ until $R$ reaches $w$, and then move $R$ to vertex $x$. When a robot $R'\neq R$ on $P$ reaches $w$, $R'$ is then routed towards an unoccupied vertex in $T_2$ by finding a root-leaf path $P'$ in $T_2$ that contains an unoccupied vertex (which must exist since $|V(T_2)| \geq k+1$), and ``pushing'' the robots on $P'$ down that path, possibly pushing (other) robots on $P'$ along the way, until $R'$ and every robot on $P'$ occupy distinct vertices in $T_2$. Once $R$ is moved to $x$, all the robots on $P$ that were moved into $T_2$ during this step are moved back into $T_1$ (i.e., we reverse the process). Afterwards, $R$ is routed to its destination in $T_2$, along the unique path from $w$ to that destination; note that this path must be devoid of robots due to the ordering imposed on the robots in the routing scheme.  The above process is then repeated for every robot in $S_2$, routing it to its location stipulated by $\iota'(S)$. 

In the last phase, we move all the robots in $S_1=S-S_2$, one by one, to $T_2$  by applying the same ``sliding'' operation as in the second phase. As a robot is moved to $T_2$, it may push the robots in $T_2$ deeper into $T_2$. Note that, during this process, no robot in $T_2$ occupies a vertex that is an ancestor of a robot $S_1$, and hence no robot in $T_2$ ``blocks'' a robot in $S_1$ from coming back to $T_1$. We then route the robots in $S_1$ from $T_2$ back into $T_1$ in the reverse order of the depths of their destinations w.r.t.~$\iota'(S)$ in $T_1$. However, if a robot has either $w$ or $x$ as a destination, then we route the robot whose destination is $w$ last, and the one whose destination is $x$ second from last.
To route a robot $R\in S_1$ from $T_2$ to its destination in $T_1$, we pull $R$ up in $T_2$ towards $w$, possibly pulling along the way other robots (that are only in $S_1$) into $T_1$,
until $R$ escapes to $x$. We then slide the robots that we pulled into $T_1$ back to $T_2$ and route $R$ to its final destination in $T_1 \cup \{w, x\}$. We repeat this step until all the robots in $S_1$ have been routed to their destination in $T_1 \cup \{w, x\}$ as stipulated by $\iota'(S)$. 

    Since $|V(H)| = \Oh(k)$, it is not difficult to verify that the total traveled length by all the robots over the whole process is $\Oh(k^3)$.  
\end{proof} \fi

\iflong \begin{longlemma} 
 \label{lem:cycle}
Let $w$ be a vertex of degree at least 3 that is contained in a (simple) cycle $C$ of length at least $2k+2$. Then $w$ is nice.
\end{longlemma}

 \begin{proof}
Let $C=(w_0=w, w_1, \ldots, w_r)$, where $r \geq 2k+1$.
Let $x \notin \{w_1, w_r\}$ be a neighbor of $w$. If $x \notin V(C)$, then the subgraphs consisting of the single edge $wx$, and the two subpaths $(w, w_1, \ldots, w_{k})$ and $(w, w_r, \ldots, w_{r-k+1})$ of $C$ satisfy conditions (i) and (ii) in Definition~\ref{def:nicevertex}, and $w$ is nice. 

Suppose now that $x \in V(C)$, and let $x_1, x_2$ be the neighbors of $x$ that appear adjacently (on either side) to $x$ on $C$.
Let $P_{wx_1}$ be the subpath of $C$ between $w$ and $x_1$ and $P_{wx_2}$ that between $w$ and $x_2$. At least one of $P_{wx_1}, P_{wx_2}$ contains at least $k+1$ vertices; otherwise, $C$ would have length at most $2k$. Assume, without loss of generality, that $P_{wx_1}$ contains at least $k+1$ vertices, and let $P_{w}^{-}$ be the subpath of $P_{wx_1}$ that starts at $w$ and that contains exactly $k+1$ vertices. Assume, without loss of generality, that $w_1$ is on $P_{w}^{-}$, and hence $w_k$ is the endpoint of $P_{w}^{-}$ other than $w$. Then it is easy to see that the three subgraphs consisting of the edge $ww_r$, the path $P_{w}^{-}$, and the connected subgraph consisting of $wx +P_{w_{r-1}w_k}$, where $P_{w_{r-1}w_k}$ is the subpath of $C$ between $w_{r-1}$ and $w_k$, satisfy conditions (i) and (ii) in Definition~\ref{def:nicevertex}, implying that $w$ is nice. 
\end{proof} \fi

\begin{figure}[htbp]
\centering
 	
	\includegraphics{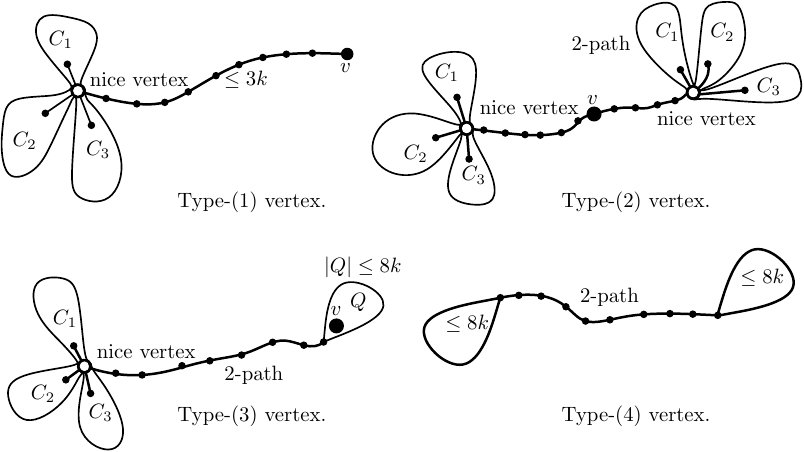}  	 
 	\caption{Illustration of the four types for a vertex $v$ that is not nice.}
	\label{fig:types}
\end{figure}

Call a path $P$ in $G$ a \emph{2-path} if $P$ is an induced path of degree-2 vertices in $G$. The next lemma is used to characterize the possible graph structures around a vertex that is not nice; an illustration for the four cases handled by the lemma is provided in Figure~\ref{fig:types}. 
 
\begin{lemma}
\label{lem:types}
For every vertex $v$ that is not nice, one of the following holds:

\begin{itemize}
    \item[(1)] There is a nice vertex at distance at most $3k$ from $v$; or
    
    \item[(2)] vertex $v$ is on a 2-path whose both endpoints are nice vertices; or 
    \item[(3)] there is a 2-path $P$ 
         such that one of its endpoints is nice and the other is contained in a connected subgraph $Q$ of size at most $8k$ that $P$ cuts from the nice vertex and such that $v$ is in $P \cup Q$; or 
        \item[(4)] Either $|V(G)| \leq 8k$, or $G$ consists of a 2-path $P$, each of whose endpoints is contained in a subgraph of size at most $8k$, such that $P$ cuts the two subgraphs from one another.  
\end{itemize}  
\end{lemma}
 
\iflong \begin{proof}
Suppose that $v$ is not a nice vertex, and assume that (1) does not hold (i.e., that there is no nice vertex at distance at most $3k$ from $v$). We will show that one of the Statements (2)-(4) must hold.

Let $T$ be a BFS-tree of $G$ rooted at $v$. For a vertex $x$ in $T$, denote by $deg_T(x)$ the degree of $x$ in $T$, by $T_x$ the subtree of $T$ rooted at $x$, and let $T_{x}^{-}= T - V(T_x)$. For two vertices $x, y \in T$, denote by $P_{xy}$ the unique path between $x$ and $y$ in $T$, and by $d_T(x, y)$ the distance between $x$ and $y$ in $T$, which is $|P_{xy}|$. 

 If there is no vertex $x$ in $T$ with $deg_T(x) \geq 3$, then $T$ is a path of vertices, each of degree at most $2$ in $T$.  Note that any edge in $G$ is between two vertices that are at most one level apart in the BFS-tree $T$. If there is an edge $xy$ in $E(G)-E(T)$ that is between two vertices that are not at the same level of $T$, then assume w.l.o.g.~that $x$ has a higher level than $y$. We can ``re-graft'' $y$ so that  $x$ becomes its parent in $T$ (i.e., cut $y$ from its current parent and attach it to $x$) to obtain another BFS-tree rooted at $v$ and containing a vertex $x$ with $deg_T(x) \geq 3$. Suppose now that all the edges in $E(G)-E(T)$ join vertices at the same level of $T$. Let $x$ and $y$ be the two vertices in $T$ of degree-3 in $G$ that are of minimum distance to $v$, and note that $xy \in E(G)$. Moreover, $P_{xy}-x - y$ is a 2-path in $G$. If both $x$ and $y$ are nice, then $v$ is contained in an induced path of degree-2 vertices in $G$ that has two nice vertices at both ends (since no edges can exist between any two vertices along that path), and hence, $v$ satisfies Statement (2) of the lemma. If one of $x, y$, say $x$ is nice, then $y$ must be nice as well by Lemma~\ref{lem:cycle}, since $d_T(v, x) =d_T(v, y) > 3k$ (as (1) does not hold), and hence $y$ is a degree-3 vertex that is contained in the cycle $P_{xy} +xy$ of length at least $3k+1$. Suppose now that both $x$ and $y$ are not nice. Observe in this case that $d_T(v,x)=d_T(v, y) \leq k+1$; otherwise, $x$ and $y$ would be nice by Lemma~\ref{lem:cycle}. If each of $T_x$, $T_y$ contains more than $k+1$ vertices, then again $x$ and $y$ would be nice. It follows that at least one of $T_x, T_y$, say $T_x$, contains at most $k$ vertices, and hence $G$ consists of a subgraph containing $T_x$, $P_{xy}$ and at most $k$ neighbors in $T_y$ of the vertices in $T_x$ that must be within distance $k$ from $y$, followed by a 2-path (the rest of $T_y$) in $G$. Therefore, $G$ consists of a subgraph of at most $4k+2$ vertices plus a 2-path in $G$, and Statement (4) in the lemma is satisfied. 

Therefore, we may assume henceforth that there is a vertex $x$ in $T$ satisfying $deg_T(x) \geq 3$.

We first claim that any subtree of $T$ that contains a vertex of degree at least $3$ in $T$ contains a vertex $x$ satisfying $deg_T(x) \geq 3$ and $|V(T_{x}^{-})| \geq 3k$, otherwise the statement of the lemma is satisfied.

Let $T_{v'}$ be a subtree containing a vertex whose degree in $T$ is at least 3, and suppose that there is no vertex $x \in T_{v'}$ satisfying $deg_T(x) \geq 3$ and $|V(T_{x}^{-})| \geq 3k$. Since there exists a vertex in $T_{v'}$ whose degree in $T$ is at least 3, for any vertex $x$ in $T_{v'}$ satisfying $deg_T(x) \geq 3$, we have $|V(T_{x}^{-})| < 3k$. Let $x$ be a vertex in $T_{v'}$ satisfying $deg_T(x) \geq 3$ and with the maximum $d_T(v, x)$. Since $|V(T_{x}^{-})| < 3k$ and $T_{x}^{-}$ is connected and contains $v$, it follows that $d_T(v, x) < 3k$, and hence $d_G(v, x) < 3k$. Since -- by our assumption -- (1) is not satisfied, there is no nice vertex within distance at most $3k$ from $v$, it follows that $x$ is not nice. Moreover, by the choice of $x$, for every child $y$ of $x$ in $T$, $T_y$ is a path of degree-2 (in $T$) vertices. Since $x$ is not nice, at most one of the paths rooted at a child of $x$ can contain more than $k$ vertices; call such a path \emph{long}. If no such long path exists, then $T_x$ contains at most $4k$ vertices; otherwise, we can split $T_x$ into two subgraphs $C_2, C_3$, intersecting only at $x$ and each containing at least $k+1$ vertices. It follows in this case that the number of vertices in $T$, and hence in $G$, is $|V(T_x)| + |V(T_{x}^{-})| \leq 7k$, and $G$ satisfies Statement (4) of the lemma (with the 2-path being empty). Suppose now that a long path $P_x$ of degree-2 vertices rooted at $x$ exists. Then since $x$ is not nice, $T_x - V(P_x)$ contains at most $k$ vertices, and hence $T$ can be decomposed into two parts $P_x$ and $T-P_x$ that intersect only at $x$, where $|V(T-P_x)| \leq 4k$. Note that, by the property of a BFS-tree, no vertex on $P_x$ at distance more than $3k$ from $x$ can be a neighbor of a vertex in $V(T-P_x)$, and hence, any vertex on $P_x$ of distance at least $3k$ from $x$ is a cut-vertex that cuts $G$ into a connected subgraph of size at most $7k$ plus a 2-path in $G$, thus satisfying the Statement (4) of the lemma. 

Suppose now that any subtree of $T$ that contains a vertex of degree at least $3$ in $T$ contains a vertex $x$ satisfying $deg_T(x) \geq 3$ and $|V(T_{x}^{-})| \geq 3k$. Since there exists a vertex $x$ satisfying $deg_T(x) \geq 3$, let $w$ be such a vertex satisfying $|V(T_{w}^{-})| \geq 3k$ and such that no ancestor $y$ of $w$ in $T$ satisfies $deg_T(y) \geq 3$ and $|V(T_{y}^{-})| \geq 3k$.

If $w$ is nice, 
consider the path $P_{wv}$ of $T$. If there exists an internal degree-3 vertex (in $T$) on $P_{wv}$, then let $x$ be the farthest such vertex from $v$ (i.e., closest such ancestor of $w$) on this path. By the choice of $w$, $|V(T_{x}^{-})| < 3k$, and hence $v \in V(T_{x}^{-})$ is contained in a subgraph of size at most $3k$ that is connected by a (possibly empty) path of vertices of degree 2 in $T$ to $w$. Since $w$ is nice and Statement (1) is not satisfied (by our assumption), we have $d_T(v, w) > 3k$. Since $|V(T_{x}^{-})| < 3k$, by the properties of BFS-trees, no vertex in $V(T_{x}^{-})$ could be adjacent in $G$ to any vertex in $T_w$ or to any vertex $y$ on $P_{x}{w}$, the path of $T$ between $x$ and $w$, satisfying 
$d_T(v, y) > 3k$. Therefore, by picking any vertex $u$ on $P_{xw}$ whose distance from $w$ is at most $1$ (Which must exist), we have a cut-vertex in $G$ that cuts $G$ into a component of size at most $6k$ containing $v$, connected by a 2-path in $G$ to a nice vertex $w$, as stated in Statement (3) of the lemma. 

Suppose now that no vertex $x$ satisfying $deg_T(x) \geq 3$ exists along $P_{vw}$. 
Note that $|P_{vw}| > 3k$,  since the distance between $v$ and the nice vertex $w$ is more than $3k$. 

If $deg_T(v) \geq 3$, then since $v$ is not nice, it follows that the union of all the subtrees rooted at the children of $v$, except the subtree containing $w$, is at most $2k$, and hence, by a similar analysis to the above, it is easy to see that $v$ is part of a connected subgraph of size at most $4k$ that is connected by a 2-path $P$ (which is a subpath of $P_{vw}$) in $G$ to a nice vertex $w$ such that $P$ cuts this subgraph from $w$, thus satisfying Statement (3) of the lemma. 

Suppose now that $deg_T(v)$ is 2, and hence $deg_G(v)=2$. Let $w'$ be the child of $v$ such that $T_{w'}$ contains $w$ and $v'$ be the other child of $v$ (i.e, such that $T_{v'}$ does not contain $w$ ).

Let $z$ be the first descendant of $v$ in $T_{v'}$ such that $deg_G(z) \geq 3$; if no such $z$ exists, then $v$ must be on a 2-path in $G$ that ends a nice vertex $w$, thus satisfying Statement (3) of the lemma.

Suppose now that $z$ exists. If $z$ is nice, then $d_T(v, z) > 3k$. Observe that any neighbor of $z$ in $G$ must be within 1 level from that of $z$ in $T$. Observe also that Since $P_{vz}$ is an induced path of degree-2 vertices in $G$. If $z$ has a neighbor $x$ in $G$ on $P_{vw}$, let $x$ be the neighbor with the minimum $d_T(v, x)$. Then $x$ must be nice since $x$ is a degree-3 vertex on a cycle of length at least $2k+2$ (see Lemma~\ref{lem:cycle}). Moreover, by the choice of $z$ and $x$, $v$ is contained in an induced path of degree-2 vertices in $G$ whose both endpoints are nice vertices ($x$ and $z$), thus satisfying Statement (2) of the lemma. Similarly, if any vertex $x$ on $P_{vw}$ has a neighbor in $G$, which must be a vertex below $z$ in $T_{v'}$, and hence $x$ must be on a cycle of length at least $2k+2$ and must be nice, and $v$ satisfies Statement (2) of the lemma. Otherwise, $v$ is contained in a 2-path in $G$ that ends at nice vertices, namely $z$ and $w$, again satisfying Statement (2) of the lemma.

Suppose now that $z$ is not nice. Since the length of $P_{zv}+P_{vw}$ is more than $3k$ and $z$ is not nice, the number of vertices in $T_z$ is at most $2k$. Note that, since $z$ is not nice, $z$ cannot be contained in a cycle of length more that $2k+1$, and hence, no edge in $G$ exists between a vertex in $T_w$ and a vertex in $T_z$. If no vertex in $P_{vw}$ has a neighbor in $T_{v'}$ then $v$ is contained in 2-path $P$ that on one side ends with a nice vertex $w$, and on the other side with a connected subgraph $T_z$ of size at most $2k$ such that $P$ cuts $T_z$ from $w$, thus satisfying Statement (3) of the lemma. Otherwise, any vertex $x$ on $P_{vw}$ with a neighbor in $T_z$ must be within distance at most $k+1$ from $v$ (otherwise, $z$ would be contained in a cycle of length at least $2k+2$ and would be nice). It follows that $v$ is contained in a connected subgraph of size at most $4k+2$ that is connected by a 2-path $P$ in $G$ to a nice vertex $w$ such that $P$ cuts $w$ from this subgraph containing $v$, and thus satisfying Statement (3) of the lemma.
  
Suppose now that $w$ is not nice. Then since $|V(T_{w}^{-})| \geq 3k$ by the choice of $w$, $|V(T_w)| \leq 2k$. 

Consider the path $P_{wv}$ of $T$. If there exists a degree-3 (in $T$) vertex on $P_{wv}$, then let $x$ be the deepest such vertex (ancestor of $w$) on this path. By the choice of $w$, $|V(T_{x}^{-})| < 3k$, and hence $v \in V(T_{x}^{-})$ is contained in a subgraph of size at most $3k$. Observe that any vertex in $V(T_{x}^{-})$ cannot have a neighbor on $P_{xw}$ that is at distance more than $3k$ from $x$. Hence, either the total size of $G$ is at most $8k$ (if a vertex in $V(T_{x}^{-})$ has a neighbor in $T_w$), or $G$ can be decomposed into a 2-path $P$, each of its endpoints is a contained in a connected subgraph of size at most $8k$ (one of the two subgraphs must contain $v$) such that $P$ separates the two. In either of the two cases, $v$ satisfies Statement (4) of the lemma.

Suppose now that no such $x$ exists, then $w$ is connected to $v$ by a path of degree-2 vertices in $T$. By symmetry, we can assume that, for every child $v'$ of $v$, either $T_v'$ does not contain a vertex $x$ satisfying $deg_T(x) \geq 3$, or contains a vertex $w'$ that is not nice and satisfying $|V(T_w')| \leq 2k$, $|V(T_{w'}^{-})| \geq 3k$, and is connected to $v$ via a path of degree-2 vertices in $T$. 

If $deg_T(v) \leq 2$, then we can argue -- similarly to the above -- that $v$ must satisfy the lemma. We summarize these arguments below for the sake of completeness. Let $w'$ be the child of $v$ such that $T_{w'}$ contains $w$ and $v'$ be the other child of $v$. Either $T_{v'}$ is a path of vertices, each of degree at most $2$ in $T$, or there exists a vertex $x$ in $T_{v'}$ such that  $deg_T(x) \geq 3$,  $x$ is not nice, $x$ satisfies $|V(T_x)| \leq 2k$ and $|V(T_{x}^{-})| \geq 3k$, and $P_{xv}$ is a path of degree-2 vertices in $T$. If $|P_{vw}| \leq 2k$, then $|V(T_{w'})| \leq 4k$ and no vertex in $T_{w'}$ can be adjacent in $G$ to any vertex in $T_{v'}$ whose depth in $T$ is more than $4k$; it is easy to see in this case that Statement (4) of the lemma holds. Suppose now that $|P_{vw}| > 2k$ and observe that since $w$ is not nice, no vertex in $T_w$ can be adjacent to a vertex in $T_{v'}$; otherwise, $w$ would be contained in a cycle of length at least $2k+2$. Therefore, the only edges in $G$ that could exist between $T_{w'}$ and $T_{v'}$ are edges between $P_{vw}$ and $T_{v'}$. If no edge exists between $T_{w'}$ and $T_{v'}$, then Statement (4) of the lemma holds. Otherwise, let $xy$ be the edge between a vertex $x$ in $P_{vw}$ and a vertex $y$ in $T{v}$ such that $x$ is of minimum depth in $T_{w'}$. If $x$ is nice then $y$ must be nice (since both would be contained in a cycle of length at least $2k+2$) and Statement (2) of the lemma holds. Suppose now that both $x$ and $y$ are not nice, which implies that each has depth at most $k+1$ in $T$. Moreover, either $T_x$ or $T_y$ must contain at most $k$ vertices. In each of these two cases, it is easy to see that Statement (4) of the lemma holds.

Suppose now that $v$ has at least three children in $T$. At most one of these children can satisfy that the subtree of $T$ rooted at it has at least $k$ vertices; otherwise, $v$ would be nice. If no child $v_i$ of $v$ satisfies that $|V(T_{v_i})| \geq k$, then it is easy to see that $|V(G)| \leq 6k$, and hence Statement (4) of the lemma is satisfied, as otherwise $v$ would be nice. Suppose now that exactly one child of $v$, say $v_1$, is such that $|V(T_{v_1})| \geq k$, and note in this case that $|V(T) -V(T_{v_1})| \leq 2k$. If $T_{v_1}$ contains a vertex $x$     satisfying $deg_T(x) \geq 3$, then $T_{v_1}$ contains a vertex $w$ that is not nice and satisfying $|V(T_w)| \leq 2k$, $|V(T_{w}^{-})| \geq 3k$, and such that $P_{vw}$ is a path of degree-2 vertices in $T$. Since $|V(T) -V(T_{v_1})| \leq 2k$, and hence, any vertex in 
$V(T) -V(T_{v_1})$ can be adjacent only to vertices on $P$ that are within distance at most $2k$ from $v$ (since $T$ is a BFS-tree), either $G$ has at most $6k$ vertices, or it consists of two connected subgraphs, each of size at most $4k$ that are connected by a 2-path of $G$ (which is a subpath of $P$) that separates the two subgraphs, and Statement (4) of the lemma holds. Suppose now that $T_{v_1}$ consists of a path $P$ of vertices each of degree at most 2 in $T$. Again, by the properties of the BFS-tree $T$, any vertex in $V(T) -V(T_{v_1})$ can be adjacent to vertices on $P$ that are within distance at most $2k$ from $v$. It follows in this case that $G$ can be decomposed into a connected subgraph of size at most $4k$ that contains $v$ and that is connected by a 2-path, and Statement (4) of the lemma holds.
\end{proof} \fi

\begin{our_definition}[Vertex Types]
\label{def:type}
Let $v$ be a vertex in $G$ that is not nice. We say that $v$ is a \emph{type-($i$) vertex}, where $i \in \{1, \ldots, 4\}$, if $v$ satisfies Statement ($i$) in Lemma~\ref{lem:types} but does not satisfy any Statement ($j$), where $j < i$ (if Statement ($j$) exists). Note that, by Lemma~\ref{lem:types}, every vertex $v$ that is not nice must be a type-($i$) vertex for some  $i \in \{1, \ldots, 4\}$. 
\end{our_definition}

The following lemma shows that we can restrict our attention to the case where there is no type-(4) vertex in $G$.

\begin{lemma}
\label{lem:specialcase}
Let $\I=(G, \R=(\M, \F), k, \ell)$ be an instance of \cmpl{}. If there exists a type-(4) vertex $v \in V(G)$ then $\I$ can be solved in time $\Oh^*((4k^2+16k)^{2k})$ and hence is \FPT.
\end{lemma}

\iflong \begin{proof}
Suppose that there exists a type-(4) vertex $v \in V(G)$. Then $G$ can be decomposed into two connected subgraphs, $H_1$, $H_2$, that are connected by a 2-path $P$ in $G$. We will show that in this case  $\I$ can be solved in \FPT-time.

Call a vertex $x \in V(G)$ \emph{critical} if $\dist_G(x, w) \leq k$ for some $w\in V(H_1) \cup V(H_2) \cup \{s_i, t_i\mid R_i\in \R\}$. That is, $x$ is critical if it is within distance at most $k$ from a vertex in one of the two subgraphs $H_1, H_2$, or if it is within distance at most $k$ from the starting position or the ending position of some robot.

Initially, all robots are located at critical vertices (since their starting vertices are critical). It is not difficult to prove, by induction on the number of moves and by normalizing the moves in a solution, that there exists an optimal solution to $\I$ in which 
if a robot $R\in \R$ is not a critical vertex, then it is moving directly without stopping or changing direction on $P$ between two critical vertices.

Define a \emph{configuration} to be an injection from $\R$ into the set of critical vertices in $G$. Define the configuration graph ${\cal C}_G$ to be the weighted directed graph whose vertex-set is the set of all configurations, and whose edges are defined as follows. There is an edge from a configuration $q$ to a configuration $q'$ if and only 
$q'$ can be obtained from $q$ via a single (parallel) move of a subset of the robots, or $q$ and $q'$ are identical with the exception of a single robot's transition from a critical vertex on $P$ to another critical vertex on $P$ along an unoccupied path; in either case, the weight of edge $(q, q')$ is the total traveled length by all the robots for the corresponding move/transition.

To decide $\I$, we compute the shortest paths (e.g., using Dijkstra's algorithm) between the initial configuration, defined based on the input instance (in which every robot is at its initial position), and every configuration in the configuration graph in which the robots in $\M$ are located at their destinations, and then choose a shortest path, among the computed ones, with total length at most $\ell$ if any exists. 

 The number of critical vertices is at most $2k\cdot 2|\R| + |V(H_1)|+|V(H_2)| \leq 4k^2 + 16k$, and hence the number of configurations is at most $(4k^2+16k)^k$.  It follows that the problem can be solved in $\Oh^*((4k^2+16k)^{2k})$ time, and hence is \FPT{}.    
\end{proof} \fi

\iflong \begin{longlemma} 
\label{lem:onerobot}
Let $\I=(G, \R=(\M, \FFF), k, \ell)$ be an instance of \cmplone{}, and let $opt(\I)$ be an optimal solution for $\I$.
Let $v$ be a type-(3) vertex, and let $Q$ be the connected subgraph of size at most $8k$ that is connected by a 2-path $P$ to a nice vertex $n_R$, and such that $v \in V(Q) \cup V(P)$. If a robot $R\in \FFF$ starts on $v$, then the route of $R$ in $opt(\I)$ cannot visit $n_R$ and then visit a vertex in $V(Q)$.
\end{longlemma}

\begin{proof}
Let $R_1 \neq R$ be the single robot in $\M$. Observe that
the only reason for $R$ to visit $n_R$ is to reverse order with $R_1$ whose destination must be in this case in $V(P) \cup V(Q)$. (Otherwise, $R$ would never have to leave $P \cup Q$.) But then $R_1$ must enter $P \cup Q$ for the last time before $R$ does, and hence $R$ will never end up in a vertex in $Q$ in $opt(\I)$ since $|P| \geq 3k$. 
\end{proof} 
 
\begin{longlemma}
\label{lem:onerobot2}
Let $\I=(G, \R=(\M, \FFF), k, \ell)$ be an instance of \cmplone{}, and let $opt(\I)$ be an optimal solution for $\I$.
Let $v$ be a type-(3) vertex, and let $Q$ be the connected subgraph of size at most $8k$ that is connected by a 2-path $P$ of length at least $k+1$ to a nice vertex $n_R$, and such that $v \in V(Q) \cup V(P)$. If a robot $R\in \FFF$ starts on $v$, then the route of $R$ in $opt(\I)$ cannot visit a vertex in $V(Q)$ and then visit $n_R$.
\end{longlemma}

  \begin{proof}
Let $R_1 \neq R$ be the single robot in $\M$. The robot $R$ might visit $Q$ because it is pushed inside by $R_1$ or by some other robot $R'\in \FFF$. If it is pushed by $R_1$ to change their order in $Q$, then it would  never need to go as far as $n_R$ (since $|P| \geq k+1$) after visiting $Q$ to make space for $R_1$. On the other hand, if it is pushed by another robot $R'$, then since neither of the two has a destination, we can exchange $R$ and $R' \in \FFF$, and change the schedule to a shorter one in which $R'$ goes to $n_R$ and $R$ ends up in $Q$ instead.
\end{proof}

\begin{longlemma} 
 \label{lem:simplify_haven}
Let $\I=(G, \R=(\M,\FFF), k, \ell)$ be an instance of \cmpl{}, $v$ a nice vertex, $H_v$ a haven for $v$, and $S$ a schedule for $\I$. There is a polynomial time algorithm that takes $S$ and computes a schedule $S'$ such that:
\begin{enumerate}
    \item each robot $R_i\in \R$ 
         enters $H_v$ at most once and leaves $H_v$ at most once; moreover, if $R_i$ starts in $H_v$, then it will never re-enter $H_v$ if it leaves it;
    \item the total traveled length of all the robots outside of $H_v$ in $S'$ is upper-bounded by the total traveled length of all the robots outside of $H_v$ in $S$; and
    \item the total traveled length of all the robots inside of $H_v$ in $S'$ is $\Oh(k^4)$. 
\end{enumerate}
\end{longlemma}

\begin{proof}
We will construct the schedule $S'$ as follows. 
For every robot $R_i\in \R$, define $v^i_{first}$ and $v^i_{last}$ as the first and last vertex, respectively, that $R_i$ visits in $H_v$. 
Let $W_i$ be the route for $R_i$ in $S$. 
We let its route $W_i'$ in $S'$ follow $W_i$ until it reaches the vertex $v^i_{first}$ for the first time, with a slight modification that we might ``freeze'' $R_i$ at its place for several time-steps when some other robot is either about to enter $H_v$ or is leaving $H_v$. 
In addition, when $R_i$ enters $H_v$ in $v^i_{first}$, we change its routing in a way that it stays in $H_v$ until it is its time to leave from $v^i_{last}$ and never returns to $H_v$ again. 
That is, when a robot $R_i$ is about to enter $v^i_{first}$, we ``freeze'' the execution of the schedule and reconfigure the robots currently in $H_v$ so that $v^i_{first}$ is empty, for example by using Lemma~\ref{lem:swap}, and then move $R_i$ to $v^i_{first}$. 
From now on, we keep $R_i$ in $H_v$, possibly moving it around each time we ``freeze'' the execution of the schedule to let a robot enter $H_v$ or leave $H_v$. 
When all the robots are either in $H_v$ or on the positions where they are at the time when $R_i$ moves from $v^i_{last}$ to some neighbor of $v^i_{last}$ and never returns to $H_v$ afterwards, we ``freeze'' the execution of the schedule and we reconfigure $H_v$ to a configuration where $R_i$ is at $v^i_{last}$ by using Lemma~\ref{lem:swap}. 
From this position $R_i$ again follows $W_i$. 
It is easy to see that the total traveled length of $S'$ outside of $H_v$ is at most the total traveled length of $S$ outside of $H_v$, as each robot is taking exactly the same steps until it reaches $H_v$ for the first time, then it stays in $H_v$ and after it leaves $H_v$, it is repeating the same steps as it does in $S$ after leaving $H_v$ for the last time. 
In addition, all the robot movements inside of $H_v$ took place only when we froze the execution of the schedule to either let a robot enter $H_v$ or leave $H_v$. This happens at most $2k$ times and each reconfiguration takes $\Oh(k^3)$ steps by Lemma~\ref{lem:swap}.
Hence the total traveled length of $S'$ inside of $H_v$ is at most $\Oh(k^4)$.
Since we also modified the solution such that every robot enters and leaves $H_v$ in $S'$ at most once, the lemma follows.\footnote{With a more careful analysis, it is possible to show that we only need at most $k$ steps to enter $H_v$ and at most $\Oh(k^2)$ to leave.}
\end{proof}
\fi

\iflong \begin{longlemma} 
 \label{lem:freebound}
$\I=(G, \R=(\M, \FFF), k, \ell)$ be an instance of \cmpl{}. 
Let $v$ be a nice vertex and $H_v$ be its haven. Then in any optimal solution $opt(\I)$ for $\I$, the total traveled length of all robots in $\R$ in $H_v$, as well as the total number of vertices of $H_v$ visited by $\R$, is $\Oh(k^4)$. 
\end{longlemma}

 \begin{proof}
    Assume, for contradiction, that the total length traveled of $opt(\I)$ inside $H_v$ is at least $c\cdot k^4 + 2k + 1$, where $c$ is the constant in big-Oh notation of Lemma~\ref{lem:simplify_haven}. Then by Lemma~\ref{lem:simplify_haven}, we can turn $opt(\I)$ into a schedule $S$ such that each robot moves exactly the same outside of $H_v$ in both $S$ and $opt(\I)$, each robot passes an edge between $H_v$ and the rest of the graph at most twice, and the total length traveled by all the robots inside $H_v$ in $S$ is at most $c\cdot k^4$. But then the total length traveled in $S$ is smaller than the total length traveled by $opt(\I)$, which is a contradiction. 
\end{proof}
\fi

\iflong 
\begin{longlemma} 
 \label{lem:niceproximity}
$\I=(G, \R=(\M, \FFF), k, \ell)$ be an instance of \cmpl{}. 
Then in any optimal solution $opt(\I)$, the total travelled length of any $R \in \FFF$ that is at distance $\lambda$ from a nice vertex is $\Oh(\lambda\cdot k + k^4)$.
\end{longlemma}

 \begin{proof}
    Assume for a contradiction that $R \in \FFF$ starts at a vertex $s_R\in V(G)$ that is at distance $\lambda$ from a nice vertex $n_R$, but the traveled length of $R$ in $opt(\I)$ is at least $2\lambda\cdot k +c_1\cdot k^3 + c_2\cdot k^4 +2k+1$, where $c_1$ is the constant in the big-Oh notation from Lemma~\ref{lem:swap} and $c_2$ is the constant in big-Oh notation from Lemma~\ref{lem:simplify_haven}.

    We can now modify $opt(\I)$ by first moving $R$ inside $H_v$, this can be done with total traveled length at most $2\lambda\cdot k + c_1\cdot k^3$ by pushing all robots along a shortest $s_R$-$n_R$ path to $H_v$ using at most $\lambda\cdot k$ steps, reconfiguring inside of $H_v$ using $c_1\cdot k^3$ steps such that the robots that we pushed in $H_v$ can leave back towards their starting positions using additional $\lambda\cdot k$ steps. Afterwards, we follow $opt(\I)$, keeping $R$ always in $H_v$ and whenever a different robot interacts with $H_v$, we reconfigure it using Lemma~\ref{lem:swap} to the position from which it needs to leave next. Let this schedule be $S$. Note that the total traveled length of $S$ outside of $H_v$ is at most $2\lambda\cdot k + c_1\cdot k^3$ larger than the total traveled length of $opt(\I)$ outside of $H_v$ and $R$ stays in $H_v$ after it entered it for the first time. By Lemma~\ref{lem:simplify_haven}, we can adapt $S$ to a schedule $S'$ such that $S'$ still does only at most $2\lambda\cdot k + c_1\cdot k^3$ more steps outside of $H_v$ than $opt(\I)$ does, but does at most $2k$ steps between $H_v$ and $G-H_v$ and at most $c_2\cdot k^4$ steps inside of $H_v$, which contradicts the optimality.  
\end{proof}
\fi

Before we are ready to prove the main theorem for this section, we first need the following lemma, which will allow us to restrict the movement of the free robots.

\begin{lemma} 
\label{lem:ball}
Let $\I=(G, \R=(\M, \FFF), k, \ell)$ be an instance of \cmpl{}.
For every $R_i \in \FFF$ that is within distance $\lambda_i$ from a nice vertex,
there is a set $B(R_i, \lambda_i)$  of vertices of cardinality $k^{\Oh(\lambda_i\cdot k +k^4)}$ such that, for any optimal solution $opt(\I)$, there exists an optimal solution $opt'(\I)$ satisfying that, in $opt'(\I)$ $R_i$ moves only on vertices in $B(R_i, \lambda_i)$ and all the remaining robots move exactly the same in $opt(\I)$ and $opt'(\I)$.
 Moreover, $B(R_i, \lambda_i)$ is computable in $\Oh(|V(G)|+|E(G)|)$ time. 
\end{lemma}

\iflong
\begin{proof}
First observe that if a vertex $v$ has degree at least $c_1\cdot k^4+k+1$, where $c_1$ is the constant from Lemma~\ref{lem:freebound}, then any connected subgraph $H_v$ of $G$ that consists of $v$ and $k^4+k+1$  
many arbitrary neighbors of $v$ is a haven. 
By Lemma~\ref{lem:freebound}, any optimal solution visits at most $c_1\cdot k^4+k$ vertices of $H_v$ and hence also at most $c_1\cdot k^4+k$ neighbors of $v$. 
This is because every robot enters $H_v$ at most once, and the total traveled length inside $H_v$ is at most $c_1\cdot k^4$. 
It follows that in any optimal schedule $opt(\I)$, any free robot $R_i$ that enters $v$ does at most one step after entering $v$. 
Otherwise, we could get a shorter schedule by moving $R_i$, after entering $v$ for the first time, directly to any neighbor of $v$ in $V(H_v)$ that is not visited by any robot. 
In addition, if $opt(\I)$ moves $R_i$ outside $H_v$ after visiting $v$, then we can modify the schedule $opt(\I)$ to obtain schedule $opt'(\I)$ in which $R_i$ enters $v$ and then moves to a neighbor $w$ of $v$ such that $w\in V(H_v)$ and $w$ is not visited by any other robot in $opt(\I)$. 
By Lemma~\ref{lem:niceproximity}, the total traveled length of any free robot $R_i\in \FFF$ is at most $c_2(\lambda_i\cdot k + k^4)$.
From the discussion above, to obtain $B(R_i,\lambda_i)$, we only need to keep vertices that can be reached from the starting position $s_i$ of $R_i$ by a path $P$ of length at most $c_2(\lambda_i\cdot k + k^4)$ such that only the penultimate vertex on $P$ can have degree larger than $c_1\cdot k^4+k$ in $G$. 
Moreover, for every vertex that has degree larger than $c_1\cdot k^4+k$, we only need to keep $c_1\cdot k^4+k$ arbitrary neighbors. 
It follows that we can obtain the set $B(R_i, \lambda_i)$ by first a modified depth-bounded BFS from $s_i$, such that for any vertex $v$ of degree at least $c_1\cdot k^4+k+1$, we do not continue the BFS from this vertex. 
The depth of this BFS is bounded by  $c_2(\lambda_i\cdot k + k^4)$ and it is easy to see that the number of vertices marked by this BFS to be included in $B(R_i, \lambda_i)$  is at most $(c_1\cdot k^4+k+1)^{1+ c_2(\lambda_i\cdot k + k^4)}$.
Afterwards, for every vertex in $B(R_i, \lambda_i)$ with degree at least $c_1\cdot k^4+k+1$, but with less than $c_1\cdot k^4+k+1$ many neighbors in $B(R_i, \lambda_i)$, we add arbitrary $c_1\cdot k^4+k+1 - |N_G(v)\cap B(R_i, \lambda_i)|$ many neighbors of $v$ to $B(R_i, \lambda_i)$. 
This concludes the construction of $B(R_i, \lambda_i)$. 
It is easy to see that $|B(R_i, \lambda_i)|\le (c_1\cdot k^4+k+1)^{1+ c_2(\lambda_i\cdot k + k^4)}\cdot (c_1\cdot k^4+k+1) = k^{\Oh(\lambda_i+k^4)}$ and the lemma follows.
\end{proof}
\fi

 \begin{theorem}
\label{thm:fptapx}
 Let $\I=(G, \R=(\M,\F), k, \ell)$ be an instance of \cmpl{}. In \FPT{}-time, we can reduce $\I$ to $p=g(k)$-many instances (for some computable function $g(k)$) $\I_1, \ldots, \I_p$ such that  $\I$ is yes-instance if and only if there exists $j\in [p]$ such that $\I_j$ is yes-instance, and for all $j\in [p]$:
 
 \begin{enumerate}
\item We can, in polynomial time, compute a schedule for 
$\I_j$ in which the total traveled length is at most $\sum_{R_i \in \M_j}\dist(s_i, t_i) + c \cdot k^5$, for some constant $c > 0$; 
\item given a schedule for $\I_j$ with total traveled length $\ell_j+d$, for some $d\in \mathbb{Z}$, we can in polynomial time compute a schedule for $\I$ of cost at most $\ell+d$; 

\item the instance  $\I_j=(G_j, \R_j=(\M_j, \F_j), k_j, \ell_j)$ satisfies that $k_j \leq k$, $\ell_j \leq \ell$, and $G_j$ is a subgraph of $G$ obtained by possibly removing degree-2 vertices from $G$; and 
    
\item  if $\M=\{R_1\}$ then $\I_j=(G_j, \R_j=(\M_j, \emptyset), k_j, \ell_j)$, where $k_j \leq k$ and $\ell_j \leq \ell$.
    Moreover, in this case there is an optimal solution $opt(\I_j)$ such that $|opt(\I_j)| = \dist_{G_{j}}(s_1, t_1)+\Oh(k^5)$, and for every robot $R_i \in \M_j \setminus \{R_1\}$, 
    the moves of $R_i$ in $opt(\I_j)$ are restricted to a vertex set of cardinality at most $k^{\Oh(k^4)}$ that is computable in linear time.    
\end{enumerate}
Furthermore, every optimal schedule for $\I$ satisfies the property that every vertex in $G$ is visited at most $\Oh(k^5)$ many times in the schedule. 
\end{theorem}
\ifshort
\begin{proof}[Proof Sketch]
    Let $\I=(G, \R=(\M, \FFF), k, \ell)$ be an instance of \cmpl{}, and let $opt(\I)$ denote an optimal schedule for $\I$. We give a nondeterministic reduction that produces a single instance $\I_j=(G_j, R_j=(\M_j, \emptyset), k_j, \ell_j)$, which can be made deterministic in \FPT{}-time.
    
    To construct $\I_j$, we will be making guesses about the initial and final segments of the robots' routes in $opt(\I)$, which may lead to redefining the starting and final positions for some of them, with the purpose of getting them close to havens -- as explained at the beginning of this section. The guessing is based on the types of the starting and final positions of the robots (see Lemma~\ref{lem:types}). We set $\ell_j=\ell$.

\smallskip
\noindent    \textbf{Case I.}\quad
    If the starting or the final position of a robot is at distance at most $11k$ from a nice vertex, we do not change it. 

\smallskip
\noindent    \textbf{Case II.}\quad
    If the starting or the final position of a robot is a type-(2) vertex on a path $P$ that connects two nice vertices, then we look separately at all the robots that start on $P$ and all the robots that end on $P$ (guessing for each free robot whether it ends on $P$).
    Let us first consider the robots that start on $P$. 
    We make a guess, for each of these robots, about whether it ever leaves $P$, and if it does, from which endpoint of $P$ it leaves for the first time. This splits the robots on $P$ into three sets $S_{\text{Left}}$, $S_{\text{Mid}}$, $S_{\text{Right}}$. Observe that these three sets have to be ``consecutive'' on $P$ (i.e., their robots respect the order of the sets), and we push the starting points of the robots in $S_{\text{Left}}$ and in $S_{\text{Right}}$ to distance at most $k$ from the nice vertex they leave from for the first time. Similarly, we split the robots that end on $P$ into the sets $T_{\text{Left}}$, $T_{\text{Mid}}$, $T_{\text{Right}}$. 
    We redefine the destinations of the robots in $T_{\text{Left}}$ and $T_{\text{Right}}$ to vertices at distance at most $k$ from their respective nice vertices through which they enter $P$, and we compute the additional travel length incurred by moving them directly from these positions at distance at most $k$ from  nice vertices to their destinations and subtract this number from $\ell_j$.
    Recall that, according to our guess, $S_{\text{Mid}}$ and $T_{\text{Mid}}$ each contains precisely the robots that never leave $P$, and hence $S_{\text{Mid}}=T_{\text{Mid}}$.
              We show that for most of the robots in $S_{\text{Mid}}$, any optimal solution takes them directly from their starting positions to their destinations, except for those that start and end at distance at most $k$ from the same nice vertex. 
    For those that do not start and end at distance at most $k$ from the same nice vertex, 
    we compute the travel length incurred by taking them directly to their destinations, subtract it from $\ell_j$, and delete these robots from $\R$.
    Finally, if $S_{\text{Mid}}\neq \emptyset$, then no robot can use $P$ to move between its two endpoints. Moreover, if a robot is on $P$, then it does not need to go beyond the $k$-closest vertices on $P$ to a nice vertex. Hence, if $S_{\text{Mid}}\neq \emptyset$, then we delete all the vertices on $P$ that are at distance at least $k+1$ from a nice vertex.

\smallskip
\noindent    \textbf{Case III.}\quad
    If the starting or the final position of a robot is a type-(3) vertex, which consists of a nice vertex $n_R$ joined by a 2-path $P$ to a vertex $u$ in a connected subgraph $C$ of size at most $8k$, then we start by treating this situation the same as that of a type-(2) vertex, with the role of the second nice vertex replaced by $u$. 
    In particular, we guess the sets $S_{\text{Left}}$, $S_{\text{Mid}}$, $S_{\text{Right}}$, $T_{\text{Left}}$, $T_{\text{Mid}}$, $T_{\text{Right}}$,
         where $S_{\text{Right}}$ (resp.~$T_{\text{Right}}$) are the robots that start their routing by going from a vertex on $P$ to $u$, or from $u$ to a vertex on $P$. 
    We also change the starting and ending positions for the robots that are not in $S_{\text{Mid}}$ to vertices that are at distance at most $k$ from either $n_R$ or $u$; we compute the routing that reflects this change, and subtract its travel length from $\ell_j$.  
    
    If $S_{\text{Mid}}\neq \emptyset$, then the vertices on which the robots in $S_{\text{Mid}}$ are positioned separate the robots that are in $S_{\text{Right}}\cup T_{\text{Right}}$ plus the robots that start or end in $C$, from the rest of the graph. Note that at this point, there is no interaction between the robots in $S_{\text{Right}}\cup T_{\text{Right}}\cup S_{\text{Mid}}$ and the other robots in $\R$. 
    Therefore, we can treat $C$ and the part of $P$ containing the robots in $S_{\text{Right}}\cup T_{\text{Right}}\cup S_{\text{Mid}}$ as a separate instance, which has the same structure as that of a type-(4) vertex. We solve this instance optimally using Lemma~\ref{lem:specialcase}, and keep on $V(P)$ only the $k$-closest vertices to $n_R$. On the other hand, if $S_{\text{Mid}}=\emptyset$ then some robots that start (resp.~end) in $S_{\text{Right}}$ (resp.~$T_{\text{Right}}$) or in $C$ might also visit $n_R$. 
    We show that, in the case of \cmplone, 
    this happens only in a very restricted setting where the single robot in $\M$ starts or ends in $C$, which can be dealt with easily with the help of Lemma~\ref{lem:specialcase}.
     On the other hand, if $|\M|\ge 2$, then this is no longer the case. We now further guess which of the robots that start (resp.~end) in $S_{\text{Right}}$ (resp.~$T_{\text{Right}}$) or in $C$ also visit $n_R$ on the other end of $P$.  Using Proposition~\ref{prop:reconfiguration}, we can show that we can always route these robots at the beginning (resp.~at the end) to (resp.~from) a position at distance at most $k$ from $n_R$ with overhead at most $\Oh(k^3)$ over an optimal routing. Afterwards, we keep in $G$ only the vertices on $P$ at distance at most $k$ from $n_R$, and remove the rest of $V(C) \cup V(P)$. 
    
    This concludes the reduction to obtain the instance $\I_j$ from $\I$.

    Let us now prove Statement 4. 
    Observe that all the free robots can now be either at distance at most $11k$ from a nice vertex and we can use Lemma~\ref{lem:ball} to compute the set $B(R_i,11k)$
    where they can move; or alternatively, they can be at distance at most $k$ from a connected subgraph of size at most $8k$ associated with a type-(3) vertex and they will stay there, and hence, their movement in this case is restricted to $9k$ vertices.
    
    We now show how to compute a schedule for $\I_j$ with total traveled length at most $\dist(s_1, t_1) + \Oh(k^5)$. To do so, we first move all the robots that are at distance at most $11k$ from a nice vertex $v$ to a haven $H_v$ of $v$. We then compute a shortest path $P$ between $s_1$ and $t_1$. We note that $P$ can only interact with a connected subgraph $C$ of a type-(3) vertex at the beginning or at the end. In both cases, we can use  Lemma~\ref{lem:specialcase} to compute an optimal routing of $R_1$ from $C$ (resp.~to) to the first (resp.~from the last) haven $P$ intersects.     
         Afterwards, all free robots on $P$ are in havens. We let $R_1$ follow $P$. Whenever $P$ interacts with a haven $H$ that contains a free robot, we replace the subpath of $P$ between the first interaction of $R_1$ with $H$ and the last interaction of $R_1$ with $H$ with a schedule computed by Lemma~\ref{lem:swap} that reconfigures $H$ taking $R_1$ between its two interaction positions. Each interaction with a haven increases the length of $P$ by $\Oh(k^3)$  and there are at most $k$ havens that contain a free robot. 
    Finally, we move the free robots at the end to their guessed destinations in $\I_j$.      We can show that this incurs a travel length of $\Oh(k^5)$. 

    For Statement 1,  if $|\M|=1$ then Statement 1 follows from Statement 4. Otherwise, all starting positions and destinations are at distance at most $11k$ from nice vertices. Due to this, we can show that we can start by incurring an $\Oh(k^2)$ travel length to route each robot to a nearest haven, and finish by incurring $\Oh(k^4)$ travel length to route all the robots from havens at distance at most $11k$ from their destinations to their destinations. We can then route the robots one-by-one from a ``starting'' haven to a ``destination'' haven along a shortest path in the same manner we routed the robot $R_1$ in the case of \cmplone{}. The routing of each robot incurs an overhead of $\Oh(k^4)$ above its shortest path length, and hence the total overhead for routing all the robots between havens is $\Oh(k^5)$ and Statement 1 follows.
\end{proof}
\fi

\iflong
\begin{proof}
Let $\I=(G, \R=(\M, \FFF), k, \ell)$ be an instance of the problem, and let $opt(\I)$ denote an optimal schedule for $\I$. The reduction will produce a set of instances 
$\I_1, \ldots, \I_p$, where $p=g(k)$,   satisfying the statement of the theorem. For ease of presentation, we will present a nondeterministic reduction that produces a single instance $\I_j=(G_j, R_j=(\M_j, \emptyset), k_j, \ell_j)$, and the reduction can be made deterministic in \FPT{}-time to produce a set of instances satisfying the statement of the theorem. 
We will assume that the reduction makes the right guesses; when we simulate the reduction deterministically, we will discard any enumeration resulting in inconsistent or incorrect guesses. To construct $\I_j$, we will be making guesses about the initial and final segments of the robots' routes in $opt(\I)$, which may lead to redefining the starting and final positions for some of them. The guessing is based on the types of the starting and final positions of the robots (see Lemma~\ref{lem:types}). Set $\ell_j=\ell$.
 Let $R$ be a robot located at a vertex $v_R$. If $v_R$ is nice then we keep $R$ at $v_r$. 

If $v_R$ is a type-(1) vertex, then we do nothing.

If $v_R$ is a type-(2) vertex, then $v_R$ is on a 2-path $P$ whose both endpoints are nice vertices, $n_R$ and $n'_R$. (For visualization purposes, think about $n'_R$ as being the ``left'' endpoint of $P$ and $n_R$ being its ``right'' endpoint.) Note that $|P| > 6k$ (otherwise, $R$ would be within distance $3k+1$ from a nice vertex and $v_R$ would be of type-(1)). Let $S$ be the set of robots that occupy vertices on $P$. We guess a partition of $S$ into three subsets $S_{\text{Left}}, S_{\text{Mid}}, S_{\text{Right}}$, where $S_{\text{Left}}$ is the subset of robots such that the first nice vertex they visit is $n'_R$, $S_{\text{Right}}$ is the subset of robots such that the first nice vertex they visit is $n_R$, and 
$S_{\text{Mid}}$ is the set of robots that do not visit any nice vertex, and hence remain on $P$. Note that $S_{\text{Left}}, S_{\text{Mid}}, S_{\text{Right}}$ may contain robots in $\FFF$. Observe that $S_{\text{Left}}$ must be the $|S_{\text{Left}}|$-many closest robots in $S$ to $n'_R$ and $S_{\text{Right}}$ must be the closest $|S_{\text{Right}}|$-many robots in $S$ to $n_R$. We route the robots in $S_{\text{Left}}$ and $S_{\text{Right}}$ as follows. For each robot $R' \in S_{\text{Left}}$ (resp.~$S_{\text{Right}}$), move it towards $n'_R$ (resp.~$n_R$), pushing along robots on its way if needed, until the first time it reaches a vertex that is at distance at most $k$ from $n'_R$ (resp.~$n_R$). 
Note that each robot in $S_{\text{Left}} \cup S_{\text{Right}}$ is now within distance $k$ from a nice vertex; let $\ell_1$ be the total traveled length during this routing. We change the starting position of each robot in $S_{\text{Left}} \cup S_{\text{Right}}$ in the instance $\I$ to its current position after the routing and subtract $\ell_1$ from $\ell_j$. 

If $v_R$ is a type-(3) vertex, let $C$ be the connected subgraph of size at most $8k$ that is connected by a 2-path $P$ to a nice vertex $n_R$, and such that $v_R \in V(C) \cup V(P)$. If $|P| \leq 3k$ then we do nothing, as in this case $v_R$ will be treated later as a type-(1) vertex. Let $u_R$ be the vertex connecting $P$ to $C$.

Suppose now that $|P| > 3k$, let $S$ be the subset of robots occupying vertices in $V(P)$ and let $S_C$ be the set of robots occupying vertices in $V(C)$. We treat this case similarly to how we treated type-(2) vertices with respect to $n_R$ and $u_R$, and guess a partitioning of $S$ into three subsets $S_{\text{Left}}, S_{\text{Mid}}, S_{\text{Right}}$, where $S_{\text{Left}}$ is 
the subset of robots such that the first vertex with degree at least three in $G$ they visit is $n_R$, $S_{\text{Right}}$ is the subset of robots such that the first vertex with degree at least three in $G$ they visit is $u_R$, and $S_{\text{Mid}}$ is the subset of robots that do not visit any nice vertex nor visit $u_R$, and hence never leave $P$. We route the robots in $S_{\text{Left}}$ and $S_{\text{Right}}$ as follows. For each 
robot $R' \in S_{\text{Left}}$ (resp.~$S_{\text{Right}}$), move it towards $n_R$ (resp.~$u_R$), pushing along robots on its way if needed, until the first time it reaches a vertex that is at distance at most $k$ from $n_R$ (resp.~$u_R$). 
  Let $\ell_2$ be the total traveled length during this routing. We redefine the starting position of each robot in $S_{\text{Left}} \cup S_{\text{Right}}$ in the instance $\I_j$ to be its current position after the routing and subtract $\ell_2$ from $\ell_j$. 
Now if $|\M| = 1$, then by Lemmas~\ref{lem:onerobot}~and~\ref{lem:onerobot2}, no free robot starts in $P$ and either visits $n_R$ and then comes back to $C$, or visits $C$ and then visits $n_R$. We note that in the case that $R_1$ either starts or ends in $C$, it is possible that free robots in $C$ cannot make way for $R_1$ to leave/enter $C$ and need to leave through $n_R$. It is not difficult to see that in this case $S_{\text{Right}}\cup S_{\text{Mid}} =\emptyset$. We can guess which robots in $C$ need to enter $n_R$ and use Lemma~\ref{lem:specialcase} to compute a schedule with the minimum total traveled length that lets precisely these robots leave $C$ toward $n_R$. Afterwards, we compute the routing of the guessed robots that takes them to vertices in $B_{\text{Left}}$ and subtract the total traveled length of this routing from $\ell_j$. We are done with this case for now. 
We note that, assuming that the guesses are correct, in all the above treatments, we maintain the equivalence between the two instances $\I$ and $\I_j$. The following treatment for the case where $|\M|>1$ may result in a non-equivalent instance as the treatment may result in an additive error of $\Oh(k^4)$ steps. 

If $|\M| > 1$, we guess which of the robots in $S_{\text{Right}}\cup S_C$ visit $n_R$ after visiting $C$ and guess the order in which they leave $C$ for the last time before they visit $n_R$ for their first time. We know that the optimal solution does route these robots to leave towards $n_R$ in this order, possibly with some additional robots entering $C$. Hence, it must be possible to reorder these robots as in their guessed order. Note that during this reordering, the remaining robots in $S_C\cup S_{\text{Right}}$ might need to move to some positions in $C$ or on the closest $k$ vertices to $C$. We guess the exact positions where they end after this reordering, and let these guessed positions be their new starting positions. By   Proposition~\ref{prop:reconfiguration}, in $\Oh(k^3)$ steps, we can reconfigure the robots that leave through $n_R$ in the correct order on $P$ and the remaining robots in $S_C\cup S_{\text{Right}}$ to their new starting positions. 
Afterwards, we can move these robots to $S_{\text{Left}}$ and shift their starting positions to be at distance at most $k$ from $n_R$ as we did above for the robots that started in $S_{\text{Left}}$. Let $\ell_3$ be the total traveled length during this rerouting. We subtract $\ell_3$ from $\ell_j$.

For a robot $R \in \M$, let $t_R$ be its destination. 

If $t_R$ is within distance $11k$ from a nice vertex (which includes $t_R$ being a type-(1) vertex), then we distinguish three cases depending on whether the current starting position $s'_R$ of $R$ is also at distance at most $11k$ from a nice vertex, or $s'_R$ is of type-(2), or $s'_R$ is of type-(3). If $s'_R$ is at distance at most $11k$ from a nice vertex, we do nothing. If $s'_R$ is of type-(2), since it is not close to a nice vertex, then our guess for $R$ was that it would never leave the path $P$ associated with this type-(2) vertex.  We treat this case as we do below as if $t_R$ were a type-(2) vertex. 
If $s'_R$ is of type-(3), since it is not close to a nice vertex, then our guess for $R$ was that it would never leave $P\cup C$, where $P$ and $C$, respectively, are the 2-path and the connected subgraph associated with this type-(3) vertex. We treat this case as we do below as if $t_R$ were a type-(3) vertex.

If $t_R$ is a type-(2) vertex, which is on a 2-path $P$ whose endpoints $n_R$ and $n'_R$ are nice, then for each $R' \in \F$, guess whether it ends on $P$ and guess its relative position on $P$ with respect to all the robots in $\R$ that end on $P$. Note that such a robot $R'$ either visits a nice vertex in $opt(\I)$, in which case its starting position is at distance at most $11k$ from a nice vertex, in which case we can assume that it ends on a vertex in $B(R',11k)$ (see Lemma~\ref{lem:ball}) and only allow the guess that $R'$ ends on $P$ if $P$ intersects $B(R',11k)$. 
Alternatively, if $R'$ starts on $P$ and never leaves it,  it will be either later removed from $\I_j$, or restricted to $k$ vertices on $P$ that are closest to either $n'_R$ or $n_R$. In either case, we can restrict the moves of $R'$ to a set $B(R')$ of vertices satisfying $|B(R')|= k^{\Oh(k^4)}$.

Now let $T$ be the set of robots that end on vertices on $P$. We guess a partitioning of $T$ into three subsets $T_{\text{Left}}, T_{\text{Mid}}, T_{\text{Right}}$, where $T_{\text{Left}}$ is the subset of robots such that the last nice vertex they visit is $n'_R$, $T_{\text{Right}}$ is the subset of robots such that the last nice vertex they visit is $n_R$, and $T_{\text{Mid}}$ is the subset of robots that do not visit any nice vertex, and hence never leave $P$. Note that $T_{\text{Mid}}$ has to be equal to the set $S_{\text{Mid}}$ that was guessed above while considering the types of the starting positions of the robots w.r.t.~$P$. 
   
Let $B_{\text{Left}}$ and $B_{\text{Right}}$ be the sequences of the $k$ closest vertices on $P_R$ to $n'_R$ and to $n_R$, respectively. For each robot $R_i \in T_{\text{Left}}$,  
let $x_i$ be the number of robots in $T_{\text{Left}} \cup T_{\text{Mid}}$ whose relative position on $P$ is farther from $n'_R$ than that of $R_i$. 
If $R_i$ has a destination $t_i$, we define a new destination $t'_i$ for $R_i$ to be the vertex on $P$ whose distance from $n'_R$ is equal to $\min\{\dist_G(n'_R, t_i), k-x_i\}$. On the other hand, if $R_i \in \FFF$, then we guess $t'_i$ to be any position at distance at most $k-x_i$ from $n'_R$ that conforms to the guessed relative position of $R_i$ with respect to the other robots on $P$.
We do the same for the robots in $T_{\text{Right}}$. 

For robots in $T_{\text{Left}}$, we now compute the routing that takes each robot $R_i$ from $t'_i$ to $t_i$ directly, if needed pushing along the (free) robots (note that this might include free robots in $T_{\text{Mid}}$) on its way. We let $\ell_3$ be the total traveled length during this computed routing and we subtract $\ell_3$ from $\ell_j$.

If $T_{\text{Mid}} \neq \emptyset$, let $R_i\in T_{\text{Mid}}\cap\M$ be a robot with starting position $s_i$ and destination $t_i$. If 
$\{s_i,t_i\}\subseteq B_{\text{Left}}$, then we do nothing with $R_i$. Notice that all robots in $S_{\text{Left}}$ must have their starting positions closer to $n'_R$ than $s_i$ and all robots in $T_{\text{Left}}$ must have their destinations closer to $n'_R$ than $t_i$; therefore, none of the changes that we did to the starting positions and the destinations of robots in $S_{\text{Left}}\cup T_{\text{Left}}$ affects $R_i$. We do the same if $\{s_i,t_i\}\subseteq B_{\text{Right}}$. 
For the remaining robots in $T_{\text{Mid}}\cap\M$, we compute the routing that takes them directly from their starting positions to their destinations, possibly pushing along free robots in $T_{\text{Mid}}$. We let the total distance traveled by all these robots in $T_{\text{Mid}}$ during this routing be $\ell_4$, we subtract $\ell_4$ from $\ell_j$, and 
we then remove these robots, as well as all the robots they pushed, from the instance. It is not difficult to see that in $opt(\I)$ the travel length incurred by these robots is precisely $\ell_4$. This is because the robots in $T_{\text{Mid}}$ have to be in a sense ordered on $P$. 
In addition, if a robot $R'$ uses $P$ only to avoid other robots passing through a nice vertex $n'_R$ (resp.~$n_R$), after $R'$ enters $P$ from $n'_R$ (resp.~$n_R$), if $R$ leaves $P$ through the same nice vertex $n'_R$ (resp.~$n_R$), then $R'$ will remain in $B_{\text{Left}}$ (resp.~$B_{\text{Right}}$). Therefore, we can start by moving any of these robots in $T_{\text{Mid}}$ that start away from $B_{\text{Left}}$ and $B_{\text{Right}}$ towards their destinations, possibly pushing some robots just outside of these buffers closer to their destinations. Then we can follow an optimal solution to move the robots in $T_{\text{Left}}$ and $T_{\text{Right}}$ to their destinations in $\I_j$ when all the robots on $P$ are in the correct final relative order, and lastly just navigate each robot directly to its destination.

Observe now that what we are left to deal with are the robots in $\FFF \cap T_{\text{Mid}}$, which appear either as the first or last sequence of robots in $T_{\text{Mid}}$ w.r.t.~$P$. 
Note that, if there is such a robot $R_i$ whose $s_i$ is not in 
 $B_{\text{Left}} \cup B_{\text{Right}}$, then we do not need to do anything since any traveled length incurred by $R_i$ has already been accounted for in the previous cases, and we can remove $R_i$ from the instance. Suppose now that $s_i$ is in one of the two buffers, say $B_{\text{Left}}$. If there is a robot in $T_{\text{Left}}$ that pushes $R_i$ outside of $B_{\text{Left}}$, then when we computed its final routing contribution to $\ell_3$, the travel length incurred by $R_i$ has been already accounted for. So now we can assume that the final position of $R_i$ is in $B_{\text{Left}}$. We guess the final position of $t_i$ of $R_i$ in $B_{\text{Left}}$. Note that the robots that we kept in $\I_j$ have their starting positions and their destinations in $B_{\text{Left}}\cup B_{\text{Right}}$. If $T_{\text{Mid}}$ was non-empty at the beginning, this means that $P$ cannot be used by any robots to go between $n'_R$ and $n_R$, and in this case we remove $V(P)\setminus(B_{\text{Left}}\cup B_{\text{Right}})$ from the graph.
 
If $t_R$ is a type-(3) vertex, let $C$ be the connected subgraph of size at most $8k$ that is connected by a 2-path $P$ to a nice vertex $n_R$, and such that $v_R \in V(C) \cup V(P)$. If $|P| \leq 3k$ then we do nothing, as in this case $t_R$ will be treated later as a type-(1) vertex. Let $u_R$ be the vertex connecting $P$ to $C$.

Similarly to how we treated the type-(2) vertices above, for each $R' \in \F$, we guess whether it ends on $P$ and guess its relative position on $P$ with respect to all the robots in $\R$ that end on $P$. If a robot in $\F$ does not end on $P$, we guess whether it ends in $C$, and in which case we also guess the exact vertex in $C$ on which the robot ends. 
Again, as in the above treatment of type-(2) vertices, $R'$ either visits a nice vertex in $opt(\I)$, in which case its starting position is at distance at most $11k$ from a nice vertex in $\I_j$, and in this case we can assume that it ends on a vertex in $B(R',11k)$ (see Lemma~\ref{lem:ball}), and only allow the guess that $R'$ ends on $C\cup P$ if $V(P\cup C)$ intersects $B(R',11k)$. Alternatively, if $R'$ starts on $P\cup C$ and never leaves it, in which case it will be either later removed from $\I_j$, restricted to $k$ vertices on $P$ that are closest to $n_R$, or restricted to $C$ and the $k$ vertices on $P$ closest to $u_R$. In either case, we can restrict the moves of $R'$ in an optimal solution to a set $B(R')$ of vertices, where $|B(R')|= k^{\Oh(k^4)}$.

Now let $T$ be the set of robots that end on $P$ and $T_C$ be the set of robots that end in $C$. We guess a partitioning of $T$ into three subsets $T_{\text{Left}}, T_{\text{Mid}}, T_{\text{Right}}$, where $T_{\text{Left}}$ is the subset of robots such that the last degree three vertex they visit is $n_R$, $T_{\text{Right}}$ is the subset of robots such that the last degree three vertex they visit is $u_R$, and $T_{\text{Mid}}$ is the subset of robots that do not visit any nice vertex, and hence never leave $P$. Note that $T_{\text{Mid}}$ has to be the same as the set $S_{\text{Mid}}$ that was guessed above while considering the types of the starting positions of the robots w.r.t.~$P$.

Let $B_{\text{Left}}$ be the sequence of the $k$ closest vertices on $P$ to $n_R$. For each robot $R_i \in T_{\text{Left}}$,  
let $x_i$ be the sum of the number of robots in $T_C$ and the number of robots in $T$ whose relative destination on $P$ is farther from $n_R$ than that of $R_i$. 
If $R_i$ has destination $t_i$, we define a new destination $t'_i$ for $R_i$ to be the vertex in $B_{\text{Left}}$ whose distance from $n_R$ is equal to $\min\{\dist_G(n'_R, t_i), k-x_i\}$. On the other hand, if $R_i$ is in $\FFF$, then we guess a destination $t'_i$ for $R_i$ from among the positions at distance at most $k-x_i$ from $n_R$ conforming to the guessed relative destinations with respect to the other robots on $P$.

For the robots in $T_{\text{Left}}$, we now compute the routing that takes each robot $R_i$ from $t'_i$ to $t_i$ directly, if needed pushing along the (free) robots (note that this might include free robots in $T_{\text{Mid}}$) on its way. We let $\ell_5$ be the total traveled length during this computed routing and we subtract $\ell_5$ from $\ell_j$.

If $T_{\text{Mid}}\neq \emptyset$, then we observe that $T_{\text{Mid}}$ separates the robots in $T_{\text{Mid}}\cup T_{\text{Right}}\cup T_C$ from the rest of the instance. This means that the remaining robots cannot use $C$ to change their order on $P$, and since their destinations are either in $V(G)\setminus(V(P)\cup V(C))$ or in $B_{\text{Left}}$, they never need to enter vertices of $V(P)\setminus B_{\text{Left}}$. (Note that in the modified instance, the destinations of the robots in $T_{\text{Left}}$ have been changed to vertices in $B_{\text{Left}}$.) If a robot in $T_{\text{Mid}}$ starts and ends in $B_{\text{Left}}$, it still might need to be pushed around by other robots in the solution; let $T'_{\text{Mid}}$ be the set of robots in $T_{\text{Mid}}\cup T_{\text{Right}}\cup T_C$ except those in $T_{\text{Mid}}$ that start and end in $B_{\text{Left}}$. We now use Lemma~\ref{lem:specialcase} to compute an optimal routing for $T'_{\text{Mid}}$ in $G[V(P)\cup V(C)]$. Note that since $P$ has length at least $3k$, there is an optimal solution that starts by moving the robots that start in $B_{\text{Left}}$ outside of $B_{\text{Left}}$, routes the robots so that the robots whose destinations in $C$ are on their correct destinations, and such that the robots with destinations on $P$ are positioned on vertices in $V(P)\setminus B_{\text{Left}}$ in the correct relative position without any robot ever entering $B_{\text{Left}}$ and finishes by routing the robots on $P$ directly to their destinations. Let $\ell_6$ be the total traveled length during this optimal routing. We subtract $\ell_6$ from $\ell_j$, remove the robots in $T'_{\text{Mid}}$ from the instance and the vertices in $V(C)\cup (V(P)\setminus B_{\text{Left}})$ from the graph.

Now, if $T_{\text{Mid}}=\emptyset$, we treat $T_{\text{Right}}$ similarly to $T_{\text{Left}}$. That is, we define $B_{\text{Right}}$ to be the sequence of the $k$ closest vertices on $P$ to $u_R$. For a robot $R_i \in T_{\text{Right}}$, we let $x_i$ be the sum of the number of robots in $T_C$ and the number of robots in $T$ whose relative destinations on $P$ are farther from $n_R$ than that of $R_i$ (i.e., closer to $u_R$). 

If $R_i$ has destination $t_i$, we define a new destination $t'_i$ for $R_i$ to be the vertex in $B_{\text{Right}}$ whose distance from $u_R$ is equal to $\min\{\dist_G(u_R, t_i), x_i+1\}$. On the other hand, if $R_i$ is in $\FFF$, then we guess a destination $t'_i$ for $R_i$ from among the positions at distance at most $x_i+1$ from $u_R$ conforming to the guessed relative destination of $R_i$ with respect to the other robots on $P$. We now compute the routing that takes each robot $R_i\in T_{\text{Right}}$ from $t'_i$ to $t_i$ directly, if needed pushing along the (free) robots (note that this might include free robots in $T_{\text{Mid}}$) on its way. We let $\ell_7$ be the total traveled length during this computed routing and we subtract $\ell_7$ from $\ell_j$.

Now if $|\M| = 1$,
then by Lemmas~\ref{lem:onerobot}~and~\ref{lem:onerobot2}, no robot starts in $P \cup C$ and either visits $n_R$ and then comes back to $C$, or visits $C$ and then visits $n_R$, and we are done with this case. We recall that, assuming that the guesses are correct, in all the above treatments for the destination vertices, we maintain the equivalence between the two instances $\I$ and $\I_j$. The following treatment for the case where $|\M|>1$ may result in non-equivalent instance, as the treatment may result in an additive error of $\Oh(k^4)$ steps.

If $|\M| > 1$, we guess which of the robots in $T_{\text{Right}}\cup T_C$ visit $n_R$ before visiting $C$ and guess the order in which they leave $n_R$ for the last time before they visit $C$; denote the set of these robots by $T_{n_R}^{C}$. For every such robot $R_i$, we now define a new destination $t''_i$ in $B_{\text{Left}}$ in the same manner as we did for the robots in $T_{\text{Left}}$ and move $R_i$ to $T_{\text{Left}}$ (and remove it from $T_{\text{Right}}\cup T_C$). 
We know that the optimal solution does route each robot $R_i$ from $t''_i$ to $t'_i$. To upper bound the contribution of such reconfiguration to the total traveled length, we observe that after all of these robots are at their positions $t''_i$ and only the remaining robots in $T_{\text{Right}}\cup T_C$ are in $V(C)\cup (V(P)\setminus B_{\text{Left}})$, then the robot $(T_{\text{Right}}\cup T_C)\setminus T_{n_R}^C$ are in $V(C)\cup B_{\text{Right}}$ (they will not need to move since we moved them when considering the starting positions in $V(C)\cup V(P)$), and we can route the robots in $T_{n_R}^C$ first to the buffer $B_{\text{Right}}$ by going directly there, maybe pushing some of the robots in $B_{\text{Right}}$ closer to $u_R$. Afterwards, we use Proposition~\ref{prop:reconfiguration} to move the robots in the correct order in $B_{\text{Right}}$ in $\Oh(k^3)$ steps. 
 Let $\ell_8$ be the total traveled length during this reconfiguration. 
 We subtract $\ell_8$ from $\ell_j$ and remove the robots left in $T_{C}\cup T_{\text{Right}}$ from the instance and the vertices in $V(C)\cup (V(P)\setminus B_{\text{Left}})$ from the graph.
 
Now, we know that the optimal solution does route these robots to leave towards $C$ in this order possibly even interacting with some robots that started in $C$, but do not end there.  Moreover, we can use the algorithm of Lemma~\ref{lem:specialcase} to reconfigure these robots in the correct order on $P$. Afterwards, we can move these robots to $S_{\text{Right}}$ and shift their starting positions to $B_{\text{Right}}$ in the same manner as we did above for the robots that started in $S_{\text{Right}}$. Let $\ell_9$ be the total traveled length during this reconfiguration. 
 We subtract $\ell_9$ from $\ell_j$.

Observe that some of the free robots still do not have destinations, but may have remained in the instance. According to our guesses, each such robot $R$ is either within a distance at most $11k$ from some nice vertex, or starts in a type-(3) vertex such that the 2-path $P$ associated with that type-(3) vertex has length at least $3k$, and never leaves $C\cup B_{\text{Right}}$, where $C$ is the connected subgraph of size at most $8k$ associated with that type-(3) vertex and $B_{\text{Right}}$ is the buffer of $k$ vertices on $P$ that are closest to $C$. By Lemma~\ref{lem:ball}, we can assume that in $opt(\I)$ the moves of every free robot $R$ that is at distance at most $11k$ from a nice vertex are restricted to $B(R,11k)$. 
Moreover, $B(R,11k)$ has cardinality $k^{O(k^4)}$ and is computable in $\Oh(|V(G)| + |E(G)|)$ time. We can guess the vertex in $B(R,11k)$ that $R$ ends up on in $opt(\I)$ and assign it as its final destination in $\I_j$. Define $B(R) = B(R,11k)$ in this case.
Now for each free robot $R$ that is in $V(C) \cup B_{\text{Right}}$ and never leaves it, we guess which vertex, out of the at most $9k$ vertices in $V(C) \cup B_{\text{Right}}$ it ends up on in the optimal solution guessed above, and assign it as its final destination in $\I_j$. Define $B(R) = V(C) \cup B_{\text{Right}}$ in this case.

Let $\I_j=(G_j, \R_j=(\M_j, \emptyset), k_j, \ell_j)$ be the resulting instance from the above. It is clear that $k_j \leq k$ since we never added any robots to the instance. Moreover, in the case where $|\M|=1$, assuming the guesses are correct, the equivalence between the two instances $\I$ and $\I_j$ follows (as argued above). Now to show in this case that for any optimal solution $opt(\I_j)$ of $\I_j$, we have $|opt(\I_j)|=\dist_{G_j}(s_1, t_1) + \Oh(k^5)$, where $R_1$ is the only robot in $\M$, let $P_1$ be a shortest path in $G_j$ between $s_1$ and $t_1$, and note that $|P_1|=\dist_{G_j}(s_1, t_1)$. Observe that every robot is now positioned at distance at most $11k$ from a nice vertex in $G_j$, or it is at a distance at most $k$ from some connected subgraph $C$ of a type-(3) vertex and will never visit the nice vertex associated with that type-(3) vertex. For every nice vertex $w$, let $H_w$ be the haven associated with $w$ as in Definition~\ref{def:nicevertex}.
First, we move all the robots within distance $11k$ from $w$ into $H_{w}$.  

Observe that $P_1$ can only interact with the connected subgraph $C$ of a type-(3) vertex either at the beginning of $P_1$ before visiting any havens, or at the end when it will no longer visit a haven afterwards. This is true since $P_1$ is a (simple) shortest path and each such connected subgraph $C$ is separated by a nice vertex from the rest of the graph. Suppose $s_1$ is at distance at most $k$ from such a connected subgraph $C$ separated by a nice vertex $n_R$; the argument is similar in the case where $t_1$ is at  distance at most $k$ from $C$. Note that all the robots that are at a distance at most $k$ from $C$ are the robots that visit $C$ before visiting $n_R$. Moreover, in an optimal solution, no such robot (other than $R_1$) visits $n_R$, as stipulated by the guess and Lemma~\ref{lem:onerobot2}. Hence, the optimal solution does reconfigure these robots in such a way that $R_1$ is the closest robot to $n_R$ among all of them, or reconfigure $R_1$ to its destination in $C$. In either of these cases, by Proposition~\ref{prop:reconfiguration},  we can perform this reconfiguration using at most $O(k^3)$ steps. After this reconfiguration (and similar one at the end of the routing), all the robots on $P_1$ are in havens. 

Let us now iteratively modify $P_1$ to a final routing of $R_1$ in which $R_1$ enters and leaves each haven at most once. Note that there are at most $k$ havens that contain at least one robot and we repeat the following process at most $k$ times each increasing the length of the routing by at most $\Oh(k^3)$. Let $H_w$ be the first haven in which $R_1$ collides with a robot on $P_1$. Let $v_{first}$ be the first vertex in $H_w$ that $R_1$ visits on $P_1$ and let $v_{last}$ be the last vertex in $H_w$ it visits on $P_1$. We first reconfigure the robots on $H_w$ such that $v_{first}$ is empty. This can be done using at most $k-1$ steps (at most $k-1$ robots each doing one step). Now we follow $P_1$ until it enters $v_{first}$. Afterwards, we use Lemma~\ref{lem:swap} to reconfigure the set of all robots that are currently in $H_w$ to obtain a resulting configuration in which $R_1$ is at $v_{last}$. The cost of this reconfiguration is $\Oh(k^3)$, and since this process is repeated at most $k$ times, the overall additional cost incurred is $\Oh(k^4)$. 

Note that technically in $\I_j$, $\M_j$ can contain other robots that $R_1$ and all these robots have destinations. However, all the robots that we did not navigate to the correct destinations are at the moment in havens, so the additional loss to the optimal solution that we did not account for to navigate $R_1$ from its destination in $\I_j$ (which is close to the haven) to its destination in $\I$ cannot be more than $\Oh(k^3)$, as it interacts in this case with at most one haven. Hence, we can assume that each of these robots is at distance at most $\Oh(k^3)$ from its destination, and we can iteratively navigate them to their destinations (which are close to havens), each incurring at most $\Oh(k^4)$ reconfiguration steps. 

It follows that there is a schedule for $R_1$ of cost $\dist_{G_j}(s_1, t_1) + \Oh(k^5)$, and the first part of Statement 4 of the theorem holds. Moreover, every robot in $\R_j$ other than $R_1$ has been assigned a final destination and its movement has been restricted to a set of vertices $B(R)$ of cardinality at most $k^{\Oh(k^4)}$, and the second part of Statement 4 of the theorem holds as well.
 
Let us now prove Statement 1 of the theorem, which states that in polynomial time we can compute a schedule for $\I_j$ with total traveled length at most $f(k)+\sum_{R_i \in \M_j}\dist(s_i, t_i)$, and given such a schedule for $I_j$, in polynomial time we can compute a schedule for $\I$ of cost at most $|opt(\I)| + 
     \delta(k)$, for some computable function $\delta$.

First notice that all the modifications to the instance, except for those performed in the treatment of type-(3) vertices, preserve the optimality of the solution. In addition, in the treatment of type-(3) vertices, the additional possible cost over that of the optimal solution is $\Oh(k^4)$ to let all the robots first leave the connected subgraphs associated with type-(3) vertices and only after all the robots that wish to leave have left, the robots that need to enter can enter. This means that whenever we have a schedule for $\I_j$ of cost $opt(\I_j) + x$, for some $x\in \mathbb{N}$, we can get a solution for $\I$ of cost  $opt(\I_j) + x+\Oh(k^4)$. This establishes the second part of Statement 1 of the theorem. Also observe that $opt(\I_j)\ge \sum_{R_i \in \M_j}\dist(s_i, t_i)$. 
Note that thanks to this treatment of type-(3) vertices,  all the starting positions and the destinations of the robots are at distance at most $11k$ from  nice vertices. We start by pushing all the robots to the closest haven to their starting position. This takes at most $\Oh(k^2)$ many steps. We temporarily choose for each robot $R_i$ an arbitrary destination $t'_i$ in the haven that is closest to its final destination $t_i$. 
Note that this increases the distance of each robot to its destination by $\Oh(k)$ many steps. Now for each robot $R_i$, we compute a shortest path $P_i$ from its current position in a haven to $t'_i$. We navigate each robot $R_i$ along $P_i$ to $t'_i$, treating the collisions in havens in exactly the same manner as we did for robot $R_1$ when proving Statement 4 of the theorem. This incurs a total cost of $\sum_{R_i \in \M_j}\dist(s_i, t_i)+ O(k^5)$. For each haven $H_w$, compute a BFS tree rooted at $w$ that spans the haven and the destinations of the robots that are in it. If two subtrees corresponding to havens $H_u$ and $H_w$ overlap, then we route all robots in $H_u$ to $H_w$ (note that the largest distance between two havens in a connected cluster is at most $22k^2$). This incurs a cost of $\Oh(k^3)$ steps overall. 
Afterwards, compute once again for the haven $H_w$ a BFS tree rooted at $w$ spanning each haven and the destinations of the robots in it. Now we can reconfigure the robots so that they are on the path $w$ to their destinations in the BFS tree and the robots on the same leaf-to-root path are in the same relative order as their destinations. This reconfiguration incurs a cost of $\Oh(k^4)$. Afterwards, each robot is at distance $\Oh(k^2)$ from its destination and we can send the robots directly to their destinations along the leaf-to-root path in the BFS tree. This incurs an additional $\Oh(k^3)$ many steps. Hence, the total traveled length by all robots is is $\sum_{R_i \in \M_j}\dist(s_i, t_i)+ O(k^5)$.

Finally, all the modifications performed while reducing $\I$ to $\I_j$, except for those $\Oh(k^4)$ steps needed in the treatment of type-(3) vertices, can be characterized as direct moves to each robot $R_i$ (without going back-and-forth) between two vertices that $R_i$ must visit in this order in $opt(\I)$. Therefore, the optimal solution on which the guess for $\I_j$ is based, does take $R_i$ from its starting position in $\I$ to its starting position in $\I_j$, then to its terminal position in $\I_j$ and finally from its terminal position in $\I$. If we just sum up these distances, and using the fact that optimal length for $I_j$ is at most $\sum_{R_i \in \M_j}\dist(s_i, t_i)+ O(k^5)$, we can conclude that $opt(\I)$ can do only at most $\Oh(k^5)$ steps that are not moving $R_i$ closer to either its start in $\I_j$, its destination in $\I_j$ or its destination in $\I$, and hence any vertex cannot be visited by a robot more than $\Oh(k^5)$ times.

Statement 2 of the theorem follows from the fact that we obtain $G_j$ is a subgraph of $G$, and
$\ell_j$ from $\ell$ by subtracting the moves the robots make to go from their starting positions in $\I$ to their starting positions in $\I'$ and the moves that they take from ending position in $\I'$ to their ending positions in $\I$, which are part of some optimal solution for $\I$. 
Statement 3 of the theorem follows from the fact that we obtain $G_j$ by only removing 2-paths in the original graph whose endpoints are also degree-2 vertices in $G$.
 \end{proof}
 \fi

The approximation algorithm claimed in the introduction (restated below) follows directly from Statements 1 and 2 of Theorem~\ref{thm:fptapx}.

\fptapxcorollary*

\section{An \FPT\ Algorithm for \cmplone\ Parameterized by $k$}
\label{sec:rushhour}
The aim of this section is to establish Theorem~\ref{thm:fptrushhour} by using Theorem~\ref{thm:fptapx}. 

\fptrushhour*

\iftrue \begin{proof}
Let $\I=(G, \R=(\M,\F), k, \ell)$ be an instance of \cmplone{}. By Statement 4 of Theorem~\ref{thm:fptapx}, in \FPT-time we can compute an equivalent instance $\I_j=(G_j, \R_j=(\M_j, \emptyset), k_j, \ell_j)$ to $\I$ satisfying that $k_j \leq k$ and $|opt(\I_j)| = \dist_{G_{j}}(s_1, t_1)+\Oh(k^5)$, where $\M=\{R_1\}$. Moreover, for every robot $R_i \in \M_j \setminus \{R_1\}$, 
    the moves of $R_i$ in $opt(\I_j)$ are restricted to the vertices of a vertex-set $B(R_i)$ of cardinality at most $k^{\Oh(k^4)}$ that is computable in linear time.

Define a state graph ${\cal Q}$ whose vertices are $k$-tuples with coordinates defined as follows. The first coordinate of a $k_j$-tuple corresponds to $R_1$ and can be any of the at most $n$ vertices in $V(G_j)$; the $i$-th coordinate of a tuple, where $i \in \{2, \ldots, k_j\}$, encodes the possible location of $R_i$ and is confined to the vertices in $B(R_i)$. We purge any tuple in which a vertex in $V(G_j)$ appears in more than one coordinate of the tuple (i.e., is occupied by more than one robot in the same time step). Since $|B(R_i)| =k^{\Oh(k^4)}$, for $i \in \{2, \ldots, k_j\}$, it follows that the number of vertices in ${\cal Q}$ is at most $\Oh(n \cdot k^{\Oh(k^5)})$. Two vertices/states $S$ and $S'$ in ${\cal Q}$ are adjacent if there is a valid (i.e., causing no collision) single (parallel) move for the robots from their locations in $S$ to their locations in $S'$; the weight of an edge $(S, S')$ in ${\cal Q}$ is the number of robots in $S$ that have moved (i.e., their positions have changed). 

Define the starting configuration $S_{start}$ in ${\cal Q}$ to be the $k_j$-tuple whose coordinates correspond to the starting positions of the robots in $\I_j$. Define $S_{final}$ to be the $k_j$-tuple whose coordinates correspond to the destinations of the robots in $\I_j$. Now compute a shortest (weighted) path from $S_{start}$ to $S_{final}$ in ${\cal Q}$ (e.g., using Dijkstra's algorithm) and accept the instance $\I$ if and only if the weight of the computed shortest path is at most $\ell_j$. The running time of the algorithm is $\Oh(|V({\cal Q})|^2)=\Oh(n^2 \cdot k^{\Oh(k^5)})$, and it is clear that the above algorithm decides the instance correctly. It follows that \cmplone{} is \FPT{} parameterized by $k$.
\end{proof} \fi

 \section{An \FPT{} Algorithm Parameterized by Treewidth and $k$}
\label{sec:treewidth}

In this section, we present an \FPT{} algorithm for {\cmpl} parameterized by the number of robots and the treewidth of the input graph combined. 
This result implies that the problem is \FPT{} on certain graph classes such as graphs of bounded outerplanarity, which include subgrids of bounded height. In this sense, the result can be seen as complementary to the recently established \NP-hardness of the problem on bounded-height subgrids~\cite{GeftHalperin}.

\iflong \subsection{Terminology and Observations}

 We begin by setting up the required terminology and proving some observations used for our algorithm.
\fi
 
\begin{our_definition}[Semi-routes]
	{\rm 
	A \emph{semi-route} for $R_i$ is a tuple $W_i=(u^i_0, \ldots, u^i_{t})$ such that each $u^i_j$ is either a vertex of $G$ or the symbol $\bot$, and such that (i) $u^i_0=s_i$ and moreover, if $R_i \in \M$, then $u^i_{t}=t_i$, and (ii) $\forall j \in [t]$, either $u^i_{j-1} =u^i_j$ or $u^i_{j-1}u^i_{j} \in E(G)$ or one out of $u^i_{j-1}$ and $u^i_j$ is the symbol $\bot$. 
 The notion of two conflicting semi-routes is identical to that for routes, except that we only consider time steps where neither semi-routes has $\bot$. 
\iflong 
Formally, two semi-routes $W_i=(u^i_0, \ldots, u^i_{t})$ and $W_j=(u^j_0, \ldots, u^j_{t})$, where $i \neq j \in [k]$, are \emph{non-conflicting} if (i) $\forall r \in \{0, \ldots, t\}$, either $u^i_r=\bot$ or $u^j_r=\bot$ or
 $u^i_r \neq u^j_r$, and (ii) $\nexists r \in \{0, \ldots, t-1\}$ such that $u^j_{r+1} =u^i_r$ and $u^i_{r+1} =u^j_r$ and $\bot\notin \{u^i_r,u^i_{r+1},u^j_{r}, u^j_{r+1}\}$. Otherwise, we say that $W_i$ and $W_j$ \emph{conflict}.\fi
	}
\end{our_definition}

Intuitively, the above definition gives us a notion of partial routes, where the robots can be thought of as having {become ``invisible'' (but still potentially in motion)} for some time steps during a route (assume only a small part of the graph is visible to us). These time steps where a robot disappears from view are represented by the symbol $\bot$. {Note that a robot can ``reappear'' at a different vertex than the one it ``disappeared'' at.} A {\em semi-schedule} for $\R$ is a set of semi-routes during a time interval $[0, t]$ that are pairwise non-conflicting.  The  \emph{length} of a semi-route and a semi-schedule is naturally defined as the number of time steps in which the robot moves from one vertex to a different vertex, and the \emph{total traveled length} of a semi-schedule is the sum of the lengths of its semi-routes.
\ifshort
 A {\em boundaried graph}~\cite{CyganFKLMPPS15} is a graph $G$ with a set $X\subseteq V(G)$ of distinguished vertices called {\em boundary vertices}; the set $X$ is called the {\em boundary} of $G$. A boundaried graph $(G,X)$  is called  a {\em $p$-boundaried} graph if $|X|\leq p$. A {\em $p$-boundaried subgraph} of  a graph $G$ is a $p$-boundaried graph $(H,Z)$ such that (i) $H$ is a vertex-induced subgraph of $G$ and (ii) $Z$ separates $V(H)\setminus Z$ from $V(G)\setminus V(H)$.  
\fi

In what follows, let $\I=(G, \R=(\M, \F), k, \ell)$ be an instance of {\cmpl}. Let $(H,Z)$ be a $p$-boundaried subgraph of $G$ with  boundary $Z$ containing all the terminals. Let $S$ be a schedule for this instance with routes $W_i, i\in [k]$.

\begin{our_definition}[Signatures of Schedules]
{\rm  
Call the tuples in the set $(Z\cup \{\uparrow,\downarrow\})^k$, {\em configuration tuples} for $(H,Z)$.
We define the {\em signature} of the schedule $S$ with respect to $(H,Z)$ as a sequence of tuples $\tau_0,\dots, \tau_t$, where each $\tau_i$ is a configuration tuple defined as follows: the $j^{th}$ coordinate of $\tau_i$ signifies whether at the end of time step $i$ \begin{itemize}\item $R_j$ is on a vertex $v\in Z$, in which case this value is $v$; or \item  whether it is inside $H$ but not on the boundary, in which case this value is $\downarrow$; or \item whether it is disjoint from $H$, in which case this value is  $\uparrow$. 
 	
 \end{itemize}
The tuple $\tau_i$ is called the {\em signature of the schedule $S$ with respect to $(H,Z)$ at time step} $i$. The tuple $\tau_0$ and $\tau_t$ are called {\em starting} and {\em ending} configuration tuples. }
\end{our_definition}

\begin{our_definition}[Checkpoints of Schedules]\label{def:affectSchedule}
{\rm 	We say that the schedule $S$:
\begin{itemize}\item  {\em externally affects $(H,Z)$ at time steps} $i$ and $i+1$ if some robot moves from a vertex outside $H$ to a vertex in $Z$ during time step $i+1$, or if some robot moves from some vertex of $Z$ to a vertex outside $H$ during time step $i+1$. That is, for some robot $R_j$, we either have that  $u^j_{i}\notin V(H)$ and $u^j_{i+1}\in Z$, or we have $u^j_{i}\in Z$ and $u^j_{i+1}\notin V(H)$. 
\item  {\em internally affects} $(H,Z)$ at time steps $i$ and $i+1$ if some robot moves from a vertex of $H$ to a vertex of $Z$ during time step $i+1$, or if some robot moves from some vertex of $Z$ to a vertex in $H$ during  time step $i+1$. 
That is, for some robot $R_j$, we have $u^j_{i} \neq u^j_{i+1}$ and moreover, we either have that $u^j_{i}\in V(H)$ and $u^j_{i+1}\in Z$ or we have $u^j_{i}\in Z$ and $u^j_{i+1}\in V(H)$.
\item  {\em affects} $(H,Z)$ at time steps $i$ and $i+1$ if it internally or externally affects $(H,Z)$ at time steps $i$ and $i+1$. 
\end{itemize}
The pairs of consecutive time steps at which the schedule affects $(H,Z)$ are called {\em checkpoints of the schedule with respect to} $(H,Z)$. 
We drop the explicit reference to $(H,Z)$ whenever it is clear from the context. 
Two checkpoints $(x,x+1)$ and $(y,y+1)$ are said to be {\em consecutive} if $x<y$ and there is no checkpoint $(z,z+1)$ such that $x<z<y$. 
 }
\end{our_definition}
Roughly speaking, the checkpoints are the time steps at which the schedule interacts in some way with the separator $Z$. The interaction involves a robot either moving to or from a vertex of $Z$. As a result, the notion of checkpoints enables one to ``decompose'' a solution along the vertices of $Z$ between consecutive checkpoints. We will argue that whenever $Z$ is sufficiently small, the number of possible checkpoints will also be small.

\iflong
\begin{longdefinition}[$v$-Rigid Pairs]\label{def:vRigidCheckpoints}
	{\rm Let $v\in Z$. A pair of configuration tuples $(\tau_1,\tau_2)$ is a {\em $v$-rigid pair} if there is a unique coordinate $j$ in which these tuples differ, exactly one out of $\tau_1[j],\tau_2[j]$ is equal to $v$ and exactly one out of $\tau_1[j],\tau_2[j]$ is equal to $\downarrow$. }
	\end{longdefinition}
	
	The above definition is useful when considering checkpoints where the schedule internally affects $(H,Z)$ with the restriction that exactly one robot moves into $Z$ from $V(H)\setminus Z$ or from $Z$ into $V(H)\setminus Z$.  

\fi 
\begin{our_definition}[Semi-schedules]\label{def:affectSemiSchedule}
{\rm 
A {\em semi-schedule with respect to} $(H,Z)$ is a semi-schedule where every vertex on every semi-route is contained in $V(H)$ and for every subtuple $(u,\bot,\dots, \bot,v)$ in a semi-route, where $u,v\in V(H)$, it must be the case that $u,v\in Z$. 
 		}
\end{our_definition}
The intuition here is that whenever the robot ``vanishes'' or ``reappears'', it happens at the boundary $Z$. This allows us to analyze how a schedule interacts with the subgraph $H$ since every movement between $H$ and the rest of $G$ must happen through the boundary $Z$.

\begin{our_definition}
	[Signatures of Semi-schedules]\label{def:signatureSemiSchedule}
	{\rm 
	The {\em signature} $\tau_0,\dots, \tau_t$ of a semi-schedule with respect to $(H,Z)$ is defined similarly to that of a schedule after making the assumption that $\bot$ denotes a vertex outside $H$, that is,  we have the symbol $\uparrow$ in the $j^{th}$ coordinate of tuple $\tau_i$ if and only if the robot $R_j$ has $u^j_i=\bot$. 
	A semi-schedule {\em externally affects} $(H,Z)$ at time steps $i$ and $i+1$ if some robot ``moves''  from $\bot$ to $Z$ during time step $i+1$ or it ``moves'' from $Z$ to $\bot$ during time step $i+1$. That is, for some robot $R_j$, we have $u^j_{i}=\bot$ and $u^j_{i+1}\in Z$ or we have $u^j_{i}\in Z$ and $u^j_{i+1}=\bot$. 
 	 The notions of internally affecting $(H,Z)$, affecting $(H,Z)$, checkpoints with respect to $(H,Z)$ are defined exactly as for schedules. }
\end{our_definition}

\begin{our_definition}[Checkpoint Tuples]
{\rm {\em 	Checkpoint tuples} of a schedule or semi-schedule $S$ with respect to $(H,Z)$ are defined as those configuration tuples for $(H,Z)$ that appear in the signature of $S$ at the time steps participating in checkpoints of $S$ with respect to $(H,Z)$. The {\em checkpoint tuple sequence} of $S$ with respect to $(H,Z)$ is the set of checkpoint tuples ordered according to their respective time steps.}
\end{our_definition}

Our dynamic programming algorithm will make use of the following observation about signatures of schedules and semi-schedules to narrow down the search space. 

\begin{observation}\label{obs:fixedStartingPartiallyFixedFinishingTuple}\label{obs:sameTuplesInBetweenAffectedTuple}
Let $\tau_0,\dots, \tau_t$ be the signature of a schedule or semi-schedule with respect to $(H,Z)$. Then, the following properties hold:

\begin{enumerate}\item  $\tau_0=(s_1,\dots,s_k)$. Moreover, the coordinates of $\tau_t$ that correspond to the robots in $\M$ are fixed, that is, $\tau_t=(t_1,\dots,t_{|\M|},v_1,\dots,v_{k-|\M|})$ for some vertices $v_1,\dots, v_{k-|\M|}$ distinct from $\{t_1,\dots,t_{|\M|}\}$.
 \item The first two and last two tuples in the signature are at checkpoints. That is, $(0,1)$ and $(t-1,t)$ are checkpoints.
	\item Let $(x,x+1),(y,y+1)$ be consecutive checkpoints such that $x+1<y$. Then, the tuples $\tau_{x+2},\dots, \tau_{y-1}$ are identical. 
	\item Let $(x,x+1)$ be a checkpoint. There is a coordinate $j$ such that  $\tau_x[j]\neq \tau_{x+1}[j]$ and at least one of $\tau_{x}[j]$ or  $\tau_{x+1}[j]$ is a vertex of $Z$.
	\item Let $(x,x+1)$ be a checkpoint. There is no coordinate $j$ such that  $\tau_x[j] =~\uparrow$ and $\tau_{x+1}[j]=~\downarrow$ or $\tau_x[j]=~\downarrow$ and $\tau_{x+1}[j]=~\uparrow$.	
	\item For every time step $x$, $\tau_x$ contains at most one occurrence of any vertex of $Z$.
	\item Let $(x,x+1)$ be a checkpoint and suppose that $\tau_x[j]\neq v\in Z$ and $\tau_{x+1}[j]=v$. 
	Then, either there is no $j'$ such that $\tau_{x}[j']=v$ or it must be the case that $\tau_{x+1}[j']\neq v$. 
	 	\item Parallel movements between vertices of the boundary $Z$ in a single time step must happen along edges that exist and moreover, no edge can be used by two robots.

\end{enumerate}

	\end{observation}

	 \begin{proof}
		The first two statements follow from the fact that the terminals are contained in $Z$. The third, fourth and fifth statements follow from the definition of checkpoints and the fact that the boundary is a separator between $V(H)\setminus Z$ and $V(G)\setminus V(H)$. The sixth and seventh statements are due to the fact that no vertex can be occupied by two robots at the same time. The eighth statement is due to the fact that a schedule or semi-schedule must be comprised of pairwise non-conflicting routes or semi-routes. 
	\end{proof}

\iflong \subsection{Using Bounded Checkpoints to Solve the Problem on Bounded Treewidth Graphs}
\fi 

\begin{observation}\label{obs:boundedSignaturesWithSmallBoundary}
Let $\I=(G, \R=(\M, \F), k, \ell)$ be an instance of {\cmpl}. Let $(H,Z)$ be a $p$-boundaried subgraph of $G$ with  boundary $Z$ containing all the terminals. 
	 If there is a schedule $S$ for $\I$ in which each vertex is {entered} by any robot
    at most $\nu(k)$ times, then the following hold.
 \begin{enumerate}\item The number of checkpoints of $S$ with respect to $(H,Z)$ is bounded by $\bigoh(pk\cdot \nu(k))$. \item The number of possible checkpoint tuple sequences of  $S$ with respect to $(H,Z)$ is bounded by $g(k,p)=p^{\bigoh(pk\cdot \nu(k))}$. 
 \end{enumerate}
\end{observation}
  \iflong \begin{proof}
From Definition~\ref{def:affectSchedule}, we have that the number of checkpoints of $S$ with respect to $(H,Z)$ is bounded by $\bigoh(pk\cdot \nu(k))$. This is because each of the $p$ vertices in $Z$ can be visited at most $\nu(k)$ times by any of the $k$ robots. Moreover, the number of possible choices for a specific checkpoint tuple is bounded by $(p+2)^{k}$. This is simply because these are configuration tuples. This gives us the stated bound on the number of possible checkpoint tuple sequences of $S$ with respect to $(H,Z)$.  
\end{proof} 
\fi 

We now proceed by defining the partial solutions which play a central role in our algorithm.

\begin{our_definition}[Partial Solutions]\label{def:partialSolutionForGivenSequence}
{\rm Let $\I=(G, \R=(\M, \F), k, \ell)$ be an instance of {\cmpl}. Let $(H,Z)$ be a $p$-boundaried subgraph of $G$ with  boundary $Z$ containing all the terminals.} 
{\rm 	Given an ordered set of configuration tuples $\gamma_1,\dots, \gamma_q$ for $(H,Z)$ where $q$ is even, $\gamma_1$ is a starting configuration tuple and $\gamma_q$ is an ending configuration tuple, a {\em partial solution corresponding to the sequence} $(\gamma_1,\gamma_2),\dots, (\gamma_{q-1},\gamma_{q})$ is a semi-schedule $S$ with respect to $(H,Z)$ such that the checkpoint tuple sequence of $S$ with respect to $(H,Z)$ is exactly these tuples in this order. That is, the checkpoints of $S$ with respect to $(H,Z)$ are $(x_1,x_2),\dots,(x_{q-1},x_q)$ and $\tau_{x_i}=\gamma_{i}$ for each $i\in [q]$. }
\end{our_definition}

\iflong 
\subsubsection{Dynamic Programming over Tree Decompositions}
\fi

\fpttw*

\ifshort
\begin{proof}[Proof Sketch.]
Let $\I=(G, \R=(\M, \F), k, \ell)$ be the given instance of {\cmpl}. 
   We start with a tree decomposition $(T,\beta)$ of $G$ of optimal width and add the terminals to every bag. Call the resulting tree decomposition $(T,\beta')$ and let its width be $w$.  By invoking Theorem~\ref{thm:fptapxcorollary} and setting $\nu(k)=\bigoh(k^{5})$ in Observation~\ref{obs:boundedSignaturesWithSmallBoundary}, we infer that the length of the checkpoint tuple sequence of $S$ with respect to each boundaried graph $(G_{x,T}^\downarrow,\beta'(x))$ where $x\in V(T)$, is at most twice the number of checkpoints in this sequence. Hence, the length of this checkpoint tuple sequence is bounded by $\lambda(k,w)=\bigoh(wk^6)$. Here, $G_{x,T}^\downarrow$ denotes the graph induced by the vertices that lie either in the bag $\beta'(x)$ or below it.
Moreover, by the same observation, the number of possible checkpoint tuple sequences of $S$ with respect to each boundaried graph $(G_{x,T}^\downarrow,\beta'(x))$ is at most $g(k,w+1)=w^{\bigoh(wk^{6})}$.

Based on this fact, our goal is to use dynamic programming to compute, for every $x\in V(T)$ and boundaried graph $(G_{x,T}^\downarrow,\beta'(x))$, and for every possible checkpoint tuple sequence of length at most $\lambda(k,w)$ with respect to this boundaried graph, the length of the best partial solution corresponding to this checkpoint tuple sequence within the boundaried graph. 
 When $x$ is the root node,  $G_{x,T}^\downarrow$ is exactly the input graph $G$. Hence, the solution is given by the minimum entry in the table computed at the root node. Note that we only describe an algorithm to compute the length of an optimal solution. However, it is straightforward to see that an optimal schedule can also be produced in the same running time. 
\end{proof}
\fi

\iflong \begin{proof}

Let $\I=(G, \R=(\M, \F), k, \ell)$ be the given instance of {\cmpl}. 
We assume that a schedule exists. This can be checked in polynomial time \cite{YuR14}. Let $\rho$ denote the upper bound on the size of an optimal schedule given by our additive approximation. If $\ell \geq \rho$, then we are done. So assume this is not the case.

We start with a nice tree decomposition $(T,\beta)$ of $G$ of optimal width. This can be computed using Proposition~\ref{fact:findtw}. Now, add the terminals to every bag. Call the resulting tree decomposition $(T,\beta')$ and let its width be $w$.  Note that although we have added the terminals to all the bags, the tree $T$ has remained the same. Hence, we just have a slightly modified version of  nice tree decompositions, where  the root bag and every leaf bag is exactly the set of terminals. However, the remaining nodes (introduce, forget, join) are as described in Definition~\ref{def:treewidth}. 
Moreover, we have the mapping $\gamma':V(T)\to 2^{V(G)}$ defined as $\gamma'(t)=\bigcup_{u\preceq t} \beta'(u)$.

By invoking Theorem~\ref{thm:fptapxcorollary} and setting $\nu(k)=\bigoh(k^{5})$ in Observation~\ref{obs:boundedSignaturesWithSmallBoundary}, we infer that the length of the checkpoint tuple sequence of $S$ with respect to each boundaried graph $(G_{x,T}^\downarrow,\beta'(x))$ where $x\in V(T)$ is at most twice the number of checkpoints in this sequence and so, is bounded by $\lambda(k,w)=\bigoh(wk^6)$. Moreover, by the same observation, the number of possible checkpoint tuple sequences of $S$ with respect to each boundaried graph $(G_{x,T}^\downarrow,\beta'(x))$ where $x\in V(T)$, is at most $g(k,w+1)=w^{\bigoh(wk^{6})}$.

Based on this fact, 
our goal is to compute, for every $x\in V(T)$ and boundaried graph $(G_{x,T}^\downarrow,\beta'(x))$, and for every possible checkpoint tuple sequence of length at most $\lambda(k,w)$ with respect to this boundaried graph, the length of the best partial solution corresponding to this checkpoint tuple sequence within the boundaried graph. For ``infeasible'' checkpoint tuple sequences, we assign them a length of  $\rho+1$ to indicate infeasibility. When $x$ is the root node,  $G_{x,T}^\downarrow$ is exactly the input graph $G$. Hence, the solution is given by the minimum entry in the table computed at the root node. Note that we only describe an algorithm to compute the length of an optimal solution. However, it will be straightforward to see that an optimal schedule can also be produced in the same running time.

Consider a node $x\in V(T)$ and let ${\cal C}_x$ denote the set of possible configuration tuples over $\beta'(x)$. That is, ${\cal C}_x=(\beta'(x)\cup \{\uparrow,\downarrow\})^k$. For every sequence $\sigma=((\tau_1,\tau_2),\dots,(\tau_s,\tau_{s+1}))$ where each $\tau_i\in {\cal C}_x$, we say that $\sigma$ is a {\em sequence at} $x$. Moreover,  $\sigma$ is called {\em good} if Properties 1-2 and Properties 4-8 listed in Observation~\ref{obs:sameTuplesInBetweenAffectedTuple} hold. Any other sequence at $x$ is {\em bad}. The length of the sequence $\sigma$ described above is $s+1$. 

For every good sequence $\sigma$ at $x$ of length at most $\lambda(k,w)$, we define $h^{x}_\sigma$ to be the value of the best partial solution corresponding to $\sigma$ in the boundaried graph $(G_{x,T}^\downarrow,\beta'(x))$. For every bad sequence $\sigma$ at $x$ of length at most $\lambda(k,w)$, $h^{x}_\sigma$ is defined as $\rho+1$. This is just to indicate that we do not look for partial solutions corresponding to these sequences as satisfying the properties in Observation~\ref{obs:sameTuplesInBetweenAffectedTuple} is a necessary condition for a sequence to correspond to checkpoint tuples. Note that for a given sequence $\sigma$ of length at most $\lambda(k,w)$, it is straightforward to check the relevant properties from Observation \ref{obs:sameTuplesInBetweenAffectedTuple} and thus determine whether it is good or bad in time $\lambda(k,w)n^{\bigoh(1)}$. 

We next proceed to the description of the computation of our dynamic programming table.  Every bad sequence $\sigma$ of length at most $\lambda(k,w)$ is identified at every bag and $h_\sigma$ is set to $\rho+1$. As standard, we next proceed bottom up, starting from the leaf nodes.

\subparagraph*{Leaf node.} At each leaf node $x$, and every good sequence $\sigma$ at $x$ of length at most  $\lambda(k,w)$ ,  $h^{x}_\sigma$ is computed by brute force since the graphs have at most $w+1$ vertices.

\subparagraph*{Forget node.}
Let $x$ be a forget node with child node $y$. Let $v$ be the unique vertex in $\beta'(y)\setminus\beta'(x)$.
Let $\sigma$ be a good sequence of length at most $\lambda(k,w)$ at the node $x$. We compute $h^{x}_\sigma$ as follows.

\begin{enumerate}
\item Let $\Sigma_v$ denote the set of all sequences obtained from $\sigma$ by going over all possible ways of substituting some occurrences of $\downarrow$ in $\sigma$ with $v$. 

Since there are at most $(w+1)\cdot \lambda(k,w)$ occurrences of $\downarrow$ in $\sigma$, the size of the set $\Sigma_v$ is bounded by $\lambda'(k,w)=2^{(w+1)\cdot \lambda(k,w)}$. 
 
The idea behind this step is to ``guess'', in the partial solution corresponding to $\sigma$, all instances where a movement is made from a vertex in $\beta'(x)$ to the forgotten vertex $v$ or from $v$ to a vertex in $\beta'(x)$.  
Once this is done, we would still need to ``guess'' all instances where a movement is made from $v$ to a vertex not in $\beta'(x)$ or vice-versa. Since $x$ is a forget node, all such movements from $v$ must go ``below'' to the vertices in $\gamma'(y)\setminus\beta'(y)$ and vice-versa. This guess is done next. 

\item Let $\Sigma'_v$ denote the set of all sequences obtained by going over every sequence $\sigma'$ in $\Sigma_v$ and inserting at most $\lambda(k,w)$ $v$-rigid pairs of tuples belonging to ${\cal C}_y$ (see Definition~\ref{def:vRigidCheckpoints}) between the elements of $\sigma'$ in all possible ways. 

 For each sequence $\sigma'$ in $\Sigma_v$, the number of resulting sequences is at most $(k\cdot (w+2)^k)^{\lambda(k,w)}$. This is because there are $2k\cdot (w+2)^k$ possible $v$-rigid pairs (fix the coordinate of $v$ and allow all possible entries in the remaining coordinates) and at most $\lambda(k,w)$ possible positions. Hence, the set $\Sigma'_v$ has size at most $\lambda''(k,w)=\lambda'(k,w)\cdot (k\cdot (w+2)^k)^{\lambda(k,w)}$.

As discussed in the last step, the goal of this step is to guess all steps where a movement is made from $v$ to a vertex not in $\beta'(x)$ or in the other direction, in which case the robot must move to or from $\gamma'(y)\setminus\beta'(y)$. This is why we have the requirement of replacing $v$ with $\downarrow$ or vice-versa in consecutive tuples, motivating Definition~\ref{def:vRigidCheckpoints}. Moreover, we only need to guess pairs of tuples where {\em only} the coordinate that has $v$ is changed, because all other pairs would already be in the sequence $\sigma$ and all occurrences of $v$ in these pairs guessed correctly in the previous step. 

\item 
Inductively, for every good sequence $\sigma'$ in $\Sigma'_v$ of length at most $\lambda(k,w)$, we would have computed $h^{y}_{\sigma'}$ at the node $y$. 

Hence, $h^{x}_\sigma$ is defined to be the mimimum $h^{y}_{\sigma'}$ where $\sigma'$ ranges over all good sequences of length at most $\lambda(k,w)$ in the set $\Sigma'_v$. 

The correctness is a consequence of the fact that our guesses in the first two steps exhaustively explore, for a partial solution $S$ corresponding to $\sigma$, all possible checkpoint tuples of $S$ with respect to $(G_{y,T}^\downarrow,\beta'(y))$.
 
\end{enumerate}

\subparagraph*{Introduce node.}
Let $x$ be an introduce node with child node $y$. Let $v$ be the unique vertex in $\beta'(x)\setminus\beta'(y)$. Let $\sigma$ be a good sequence of length at most $\lambda(k,w)$ at the node $x$. We compute $h^{x}_\sigma$ as follows.

\begin{enumerate}\item  Let $\mu$ denote the total number of moves made in $\sigma$ from $v$ to vertices of $\beta'(y)$ or vice-versa. That is, $\mu$ denotes the number of pairs $(\tau_i,\tau_{i+1})$ in $\sigma$ such that for some $j$, either $\tau_i[j]=v$ and $\tau_{i+1}[j]=v'$ for some $v'\in \beta'(y)$ or $\tau_i[j]=v'$ and $\tau_{i+1}[j]=v$ for some $v'\in \beta'(y)$. 

The idea behind this step is to separate the required partial solution corresponding to $\sigma$ into two parts -- one part for all the moves between $v$ and $\beta'(y)$ and the second part will then be a semi-schedule that is a partial solution corresponding to an appropriate sequence at the child node $y$. Notice that there cannot be a move made directly from $v$ to vertices of $\gamma'(y)\setminus\beta'(y)$ because $x$ is an introduce node and so, $v$ is separated from $\gamma'(y)\setminus\beta'(y)$ by $\beta'(y)$.

	\item  Let $\sigma'$ denote the sequence obtained from $\sigma$  by replacing every occurrence of $v$ with $\uparrow$ and remove all pairs of tuples that are identical.
The idea in this step is to focus on the part of the partial solution corresponding to $\sigma$ (call it $S$) that lies in the boundaried graph $(G_{y,T}^\downarrow,\beta'(y))$, which is a semi-schedule (call it $S'$) with respect to $(G_{y,T}^\downarrow,\beta'(y))$. Since $v$ is a vertex outside $(G_{y,T}^\downarrow)$, all occurrences of $v$ must be treated as an ``external'' vertex and so, we replace these occurrences with $\uparrow$. If this results in a pair $(\tau_i,\tau_{i+1})$ where the tuples are identical, then this pair cannot be a pair of checkpoint tuples of $S'$ with respect to $(G_{y,T}^\downarrow,\beta'(y))$ and so, we remove these.

	\item Finally, $h^{x}_\sigma$ is defined to be $\mu$ plus the value of $h^{y}_{\sigma'}$ at the node $y$.
	\end{enumerate}

\subparagraph*{Join node.}
Let $x$ be a join node with children nodes $y_1$ and $y_2$.
Let $\sigma$ be a good sequence of length at most $\lambda(k,w)$ at the node $x$. 
Let $\mu$ denote the total number of moves made in $\sigma$ between vertices of $\beta'(y)$.
That is, $\mu$ denotes the number of tuples $(u,v,\tau_i,\tau_{i+1},j)$ where $u,v\in \beta'(y)$, 
$(\tau_i,\tau_{i+1})$ is a pair in $\sigma$, $\tau_i[j]=u$ and $\tau_{i+1}[j]=v$.

We now compute $h^{x}_\sigma$ as follows. 
\begin{enumerate}
\item Let $\Sigma'$ denote the set of all sequences obtained from $\sigma$ by going over every occurrence of $\downarrow$ in $\sigma$, and replacing it with one out of $\downleftarrow$ or $\downrighttarrow$. 

The set $\Sigma'$ has size at most $\lambda'(k,w)$ by the same reasoning as that used in the first step of processing forget nodes.

 The idea in this step is to ``guess'', for every occurrence of a vertex of $\gamma'(x)\setminus\beta'(x)$ in $\sigma$ (which is indicated by $\downarrow$), whether this vertex is in $\gamma'(y_1)\setminus\beta'(y_1)$ (denoted by $\downleftarrow$) or in $\gamma'(y_2)\setminus\beta'(y_2)$ (denoted by $\downrighttarrow$) in some partial solution $S$ corresponding to $\sigma$. 
 
\item Let $\Sigma''$ denote pairs  $(\sigma_1,\sigma_2)$ of sequences obtained as follows.
For every $\sigma'\in \Sigma'$ do:
\begin{enumerate}\item  In $\sigma'$, replace all occurrences of $\downleftarrow$ with $\downarrow$ and replace all occurrences of $\downrighttarrow$ with $\uparrow$. Call the result $\sigma_1$.  \item In $\sigma'$, replace all occurrences of $\downleftarrow$ with $\uparrow$ and replace all occurrences of $\downrighttarrow$ with $\downarrow$. Call the result $\sigma_2$. 
\item Add the pair $(\sigma_1,\sigma_2)$ to $\Sigma''$.

\end{enumerate}

The idea in this step is to generate, for every guess made in the first step, two sequences $\sigma_1$ and $\sigma_2$, expressing the checkpoint tuple sequences of $S$ with respect to the boundaried graphs $(G_{y_1,T}^\downarrow,\beta'(y_1))$ and $(G_{y_2,T}^\downarrow,\beta'(y_2))$, respectively.    Hence, to obtain $\sigma_1$, the vertices ``guessed'' to be below $\beta'(y_2)$ are denoted by $\uparrow$ 
and the vertices ``guessed'' to be below $\beta'(y_1)$ are denoted by $\downarrow$. Similarly, to obtain $\sigma_2$, the vertices ``guessed'' to be below $\beta'(y_1)$ are denoted by $\uparrow$ 
and the vertices ``guessed'' to be below $\beta'(y_2)$ are denoted by $\downarrow$.

We next use the fact that the length of $S$ can be expressed as the sum of the lengths of partial solutions $S_1$ and $S_2$ corresponding to $\sigma_1$ and $\sigma_2$ minus the number of movements made between vertices of $\beta'(x)$ to avoid double counting these moves.  

\item Hence, $h^{x}_\sigma$ is defined to be the minimum over all $h^{y_1}_{\sigma_1}+h^{y_2}_{\sigma_2}-\mu$ where we range over all pairs  $(\sigma_1,\sigma_2)\in \Sigma''$ where both $\sigma_1$ and $\sigma_2$ are good sequences. If such $\sigma_1$ and $\sigma_2$ do not exist, then $h^{x}_\sigma$ is set to $\rho+1$ to indicate infeasibility.

\end{enumerate}

The number of entries at each bag of the tree decomposition as well as the time required to compute the table at each bag is clearly FPT in $k$ and $w$. This completes the proof of Theorem~\ref{thm:boundedTreewidthAlgorithm}. 
\end{proof} \fi

\section{Using the Total Energy as the Parameter}
In the final part of our paper, we consider the complexity of \cmpl\ in settings where we are dealing with many robots but the energy budget $\ell$ is small. While intuitively this may seem like an ``easier'' case due to the fact that $\ell$ immediately bounds the number of robots that will end up moving in a hypothetical schedule, we show that the problem (and even its restriction to the often-studied case where all robots have destinations) is unlikely to be fixed-parameter tractable when parameterized by $\ell$.

 \wone*

\iftrue \begin{proof}
We provide a parameterized reduction from the well-known \W[1]-complete \textsc{Multicolored Clique} problem~\cite{DowneyFellows13,CyganFKLMPPS15}: decide whether a given $\kappa$-partite graph $(V_1,\dots,V_\kappa,E)$ contains a clique of size $\kappa$. To avoid any confusion, we explicitly remark that in this proof we use $\kappa$ to denote the parameter of the initial instance of \textsc{Multicolored Clique} and not the number of robots in the resulting instance of \cmpl.

Our reduction takes an instance $(V_1,\dots,V_\kappa,E)$ of \textsc{Multicolored Clique} and constructs an instance $(G,\R=(\M,\F),\kappa',\ell)$ of \cmpl\ as follows. To obtain $G$, we: 
\begin{enumerate}
\item subdivide each edge $e\in E$ $\kappa^3$ many times; 
\item attach a single new pendant vertex $v'$ to each vertex $v\in V_1\cup \dots \cup V_\kappa$ in the original graph; and
\item for each pair of integers $i,j$ such that $1\leq i< j\leq \kappa$, construct a new vertex $s_{i,j}$ and make it adjacent to every vertex in $V_i$, and construct a new vertex $t_{i,j}$ and make it adjacent to every vertex in $V_j$.
\end{enumerate}
For $\R=(\M, \F)$, we set $\F=\emptyset$ and add two sets of robots into $\M$. First, for each vertex $v$ in $V_1\cup \dots \cup V_\kappa$, we construct a \emph{blocking} robot $r^v$ and set $v$ as its starting point as well as its destination.
Second, for each of the newly constructed $s_{i,j}\in V(G)$, we construct a \emph{clique} robot $r^{i,j}$, set $s_{i,j}$ as its starting point and $t_{i,j}$ as its destination. Finally, we set $\ell:=2\kappa+{\kappa \choose 2}\cdot (\kappa^3+3)$ and $\kappa':=|\R|$. This completes the description of our reduction.

Towards proving correctness, we first observe that a hypothetical schedule for $(G,\R,\kappa',\ell)$ can never use less than $\ell:=2\kappa+{\kappa \choose 2}\cdot (\kappa^3+3)$ travel length. Indeed, there are ${\kappa \choose 2}$ many clique robots, and the shortest path between their starting and destination points has length $\kappa^3+3$. Moreover, every schedule must route each clique robot $r^{i,j}$ through at least one vertex in $V_i$ and at least one vertex in $V_j$, and hence a hypothetical solution must spend a length of at least $2$ to move at least one of the blocking robots in each $V_p$, $p\in [\kappa]$, out of the way and then back to its initial position (which is also its destination).

Now, assume that $(G,\R,\kappa',\ell)$ is a yes-instance. By the argument above, this implies that there is a schedule which moves precisely one of the blocking robots in each $V_p$, $p\in [\kappa]$; let us denote by $w_p$ the starting vertex of the unique blocking robot which is moved from $V_p$. Since $(G,\R,\kappa',\ell)$ is a yes-instance, $E$ must contain an edge between each $w_i$ and $w_j$, $1\leq i<j\leq \kappa$; in particular, these vertices induce a clique in $(V_1,\dots,V_\kappa,E)$.

For the converse, assume that $(V_1,\dots,V_\kappa,E)$ contains a $\kappa$-clique $w_1,\dots,w_\kappa$. We construct a solution for $(G,\R,\kappa',\ell)$ as follows. First, for each robot starting at $w_i$, $i\in [\kappa]$, we spend a length of 1 to move it to the pendant vertex $w_i'$. Next, we route each robot $r^{i,j}$ from $s_{i,j}$ to $w_i$, through the $(\kappa^3+1)$-length path to $w_j$ (whereas the existence of this path is guaranteed by the existence of the edge $w_iw_j\in E$), and then to its destination $t_{i,j}$. Finally, we route each blocking robot which originally started at $w_i$ back to $w_i$ from $w'_i$. The travelled length of this schedule is precisely $2\kappa+{\kappa \choose 2}\cdot (\kappa^3+3)$, as desired.
\end{proof} \fi

As mentioned in the introduction, we complement Theorem~\ref{thm:wone} with a fixed-parameter algorithm on graph classes of bounded local treewidth.

Before proceeding to the proof, we recall the syntax of Monadic Second Order Logic (MSO) of graphs. It contains the logical connectives $\vee,$ $\land,$ $\neg,$ 
$\Leftrightarrow,$ $\Rightarrow,$ variables for vertices, edges, sets of vertices and sets of edges, and the quantifiers $\forall$ and $\exists$, which can be applied to these variables. It also contains the following binary relations: (i) $u\in U$, where $u$ is a vertex variable and $U$ is a vertex set variable; (ii) $d \in D$, where $d$ is an edge variable and $D$ is an edge set variable;
(iii) $\mathbf{inc}(d,u),$ where $d$ is an edge variable, $u$ is a vertex variable, with the interpretation  that the edge $d$ is incident to $u$; 
(iv)  equality of variables representing vertices, edges, vertex sets and edge sets. 
 
 The well-known Courcelle's theorem \cite{Courcelle90,ArnborgLS91} states that checking whether a given graph $G$ models a given MSO formula $\phi$ can be done in {\FPT} time parameterized by the treewidth of $G$ and the size of $\phi$. Morover, this result holds even if some vertices of $G$ are given labels or colors (i.e., we allow a fixed number of additional unary relations over $V(G)$). This is because one can produce an equivalent graph $G'$ such that $G$ has bounded treewidth if and only if $G'$ does, plus an alternate MSO formula $\phi'$ such that $G$ models $\phi$ if and only if $G'$ models $\phi'$.

\fptenergy*

\iftrue
\begin{proof}
	Let $\I=(G, \R=(\M, \F), k, \ell)$ be the given instance of {\cmpl} where $G$ belongs to a graph class that is closed under taking subgraphs and in which the treewidth of a graph of diameter $d$ is bounded by $f(d)$ for some function $f$. Let $\M'$ denote the set of those robots in $\M$ whose final position is distinct from the starting position. Since each of the robots in $\M'$ must move at least one step, it follows that if $|\M'|>\ell$, then the instance is a no-instance. So, assume that $|\M'|\leq \ell$. Let $U$ denote the set of starting points of the robots in $\M'$. 

 Consider an optimal schedule $S$. We say that an edge $e$ appears in a route $W=(u_0, \ldots, u_{t})$ in $S$ if the endpoints of $e$ are the vertices $u_{j-1}$ and $u_{j}$ for some $j\in [t]$. We say that $e$ appears in $S$ if it appears in some route in $S$. Let $E_S$ denote the set of edges of $G$ that appear in $S$. We observe that every connected component of the subgraph of $G$ induced by $E_S$ (call it $G_S$) contains a vertex of $U$. If this were not the case and there is a connected component of $G_S$ disjoint from $U$, then we can delete all the routes whose edges appear in this connected component and still have a feasible schedule for the given instance, contradicting the optimality of $S$.  Moreover, since $S$ has total energy at most $\ell$, it follows that at most $\ell$ edges can appear in $S$. Hence, it follows that $S$ is contained in the subgraph $H$ defined as the union of the subgraphs induced by the $\ell$-balls centered at the starting positions of the robots in $\M'$.

Notice that since $H$ is obtained by taking the union of at most $\ell$ graphs, each of diameter at most $2\ell$, it follows that $H$ has diameter at most $2\ell^2$ and so, the treewidth of $H$ is upper bounded by $f(2\ell^2)$. It remains to argue that {\cmpl} is FPT parameterized by the energy and treewidth of the input graph. We do this by making some guesses, obtaining a bounded number of MSO formulas such that input is a yes-instance if and only if at least one of these formulas is true on the graph $H$ with a label function (with constant number of labels) and then invoking Courcelle's theorem~\cite{Courcelle90,ArnborgLS91}. Let $\M''$ denote the subset of $\M\setminus \M'$ contained in $H$ and let $\R'=\R\cap V(H)$.

	Suppose that $\alpha\leq \ell$ robots move in $S$ and let the routes in $S$ be $\{W_i=(v^i_{1},\dots,v^i_{d_i})\mid i\in [\alpha]\}$. Assume that the first $|\M'|$ routes, $W_1,\dots,W_{|\M'|}$, correspond to the robots in $\M'$.   We guess $\alpha$ and $d_1,\dots, d_\alpha$ as these are all at most $\ell$. Moreover, we guess, for  every pair of vertices $v^i_j$ and $v^p_q$ whether they are equal. Let this be expressed by a function $\mu(v^i_j,v^p_q)$ that evaluates to 1 if and only if these vertices are guessed to be equal. 
	Thus, even though we do not know the actual vertices in the solution (except for the starting and endpoint positions of the robots in $\M'$), this guess gives us the ``structure'' of the entire solution. Notice that the number of guesses is bounded by a function of $\ell$.
		Now, we consider the labeled graph $H$ where the starting points of all robots in $\R'\setminus \M''$ are labeled red and the starting points of the robots in $\M''$ are labeled blue. Fix a guess of $\alpha$, the numbers $d_1,\dots,d_\alpha$ and a guess of the function $\mu$ and express the following as an MSO formula.  

There are ordered vertex sets $\{W_i=(v^i_{1},\dots,v^i_{d_i})\mid i\in [\alpha]\}$ such that the following hold:
\begin{enumerate} \item Each $W_i$ is a walk in $H$.
\item  For every $i,j,p,q$, $v^i_j$ and $v^p_q$ are equal if and only if $\mu(v^i_j,v^p_q)$ is 1.
\item  The walks $W_i$ are pairwise non-conflicting routes.
 \item  For each of the first $|\M'|$ walks $W_1,\dots,W_{|\M'|}$, the initial and final vertices are precisely the starting and ending points, respectively, of some robot in $\M'$.
\item  For each labeled vertex appearing on any walk $W_i$, this vertex must be the first vertex of some walk $W_i$.
\item For each walk $W_i$ starting at a blue vertex, it must terminate at the same blue vertex. 
 	
\end{enumerate}

Expressing the above in MSO is straightforward to do, so we do not go in to lower level details of the MSO formula.

We next argue that the input is a yes-instance if and only if there is a guess of $\alpha$,  the numbers $d_1,\dots,d_\alpha$ and a guess of the function $\mu$ such that the resulting MSO formula is true on $H$. Consider the forward direction. An optimal schedule $S$ with routes $W_1,\dots,W_\alpha$ naturally gives us a ``correct'' guess of $\alpha$,  $d_1,\dots,d_\alpha$ and the function $\mu$. Moreover notice that Properties 1-4 are clearly satisfied. For Property 5, notice that only the starting points of robots that never move can be disjoint from the set of initial vertices of the routes $\{W_i\mid i\in [\alpha]\}$ in the optimal schedule $S$. Since these robots do not move, they also cannot appear on any route $W_i$ or they would obstruct a robot that moves. Finally, Property 6 is satisfied since every robot in $\M''$ that moves also ends up returning to its starting point. The converse is trivial.
   \end{proof}
\fi

 \section{Conclusion}
 \label{sec:conclusion}
 In this paper, we presented parameterized algorithms for fundamental coordinated motion planning problems on graphs. In particular, we proved that \cmplone is fixed-parameter tractable parameterized by the number of robots, and that \cmpl is fixed-parameter tractable parameterized by the number $k$ of robots and the treewidth of the input graph combined. This latter result implies that \cmpl{} is fixed-parameter tractable in several graph classes which may be of interest w.r.t.~application domains, including graphs of bounded outerplanarity. 

 We conclude by highlighting three directions that may be pursued in future works:
  
 \begin{enumerate}
\item Is \cmpl{} \FPT{} parameterized by $k$ (alone) on planar graphs, or even on general graphs?

\item What is the complexity of \cmpl{} on trees? While \cmplone\ is known to be in $\Pol$ on trees, the complexity of the general problem on trees remains unresolved.

\item The focus of this paper was on minimizing the travel length/energy, which is one of the most-studied optimization objectives. Nevertheless, there are other important objectives, most notably that of minimizing the makespan. It seems that at least in the settings considered here, optimizing the makespan may be a computationally more challenging problem.
 In particular, the question of whether  \textsc{Coordinated Motion Planning} on trees is \FPT{} parameterized by $k$ remains open.

   \end{enumerate}

\bibliography{ref}

\end{document}